\documentclass[journal]{vgtc}

\usepackage{graphicx}

\usepackage{microtype}
\PassOptionsToPackage{warn}{textcomp}
\usepackage{textcomp}
\usepackage{mathptmx}
\usepackage{times}

\usepackage{cite}
\usepackage{tabu}
\usepackage{booktabs}

\title{Visualization and Selection of Dynamic Mode Decomposition Components for Unsteady Flow}

\author{Tim Krake$^{1,2}$, Stefan Reinhardt$^{1,2}$, Marcel Hlawatsch$^{1}$, Bernhard~Eberhardt$^{2}$, and~Daniel~Weiskopf$^{1}$
\\
\\
 $^{1}$VISUS, University of Stuttgart, Germany\\
 $^{2}$Hochschule der Medien, Germany.
}

\abstract{
Dynamic Mode Decomposition (DMD) is a data-driven and model-free decomposition technique.
It is suitable for revealing spatio-temporal features of both numerically and experimentally acquired data.
Conceptually, DMD performs a low-dimensional spectral decomposition of the data into the following components:
The modes, called DMD modes, encode the spatial contribution of the decomposition, whereas the DMD amplitudes specify their impact.
Each associated eigenvalue, referred to as DMD eigenvalue, characterizes the frequency and growth rate of the DMD mode. 
In this paper, we demonstrate how the components of DMD can be utilized to obtain temporal and spatial information from time-dependent flow fields.
We begin with the theoretical background of DMD and its application to unsteady flow. 
Next, we examine the conventional process with DMD mathematically and put it in relationship to the discrete Fourier transform. 
Our analysis shows that the current use of DMD components has several drawbacks.
To resolve these problems we adjust the components and provide new and meaningful insights into the decomposition:
We show that our improved components describe the flow more adequately.
Moreover, we remove redundancies in the decomposition and clarify the interplay between components, allowing users to understand the impact of components.
These new representations ,which respect the spatio-temporal character of DMD, enable two clustering methods that segment the flow into physically relevant sections and can therefore be used for the selection of DMD components. 
With a number of typical examples, we demonstrate that the combination of these techniques allow new insights with DMD for unsteady flow.
} 


\teaser{
	\centering
	\includegraphics[width=\linewidth]{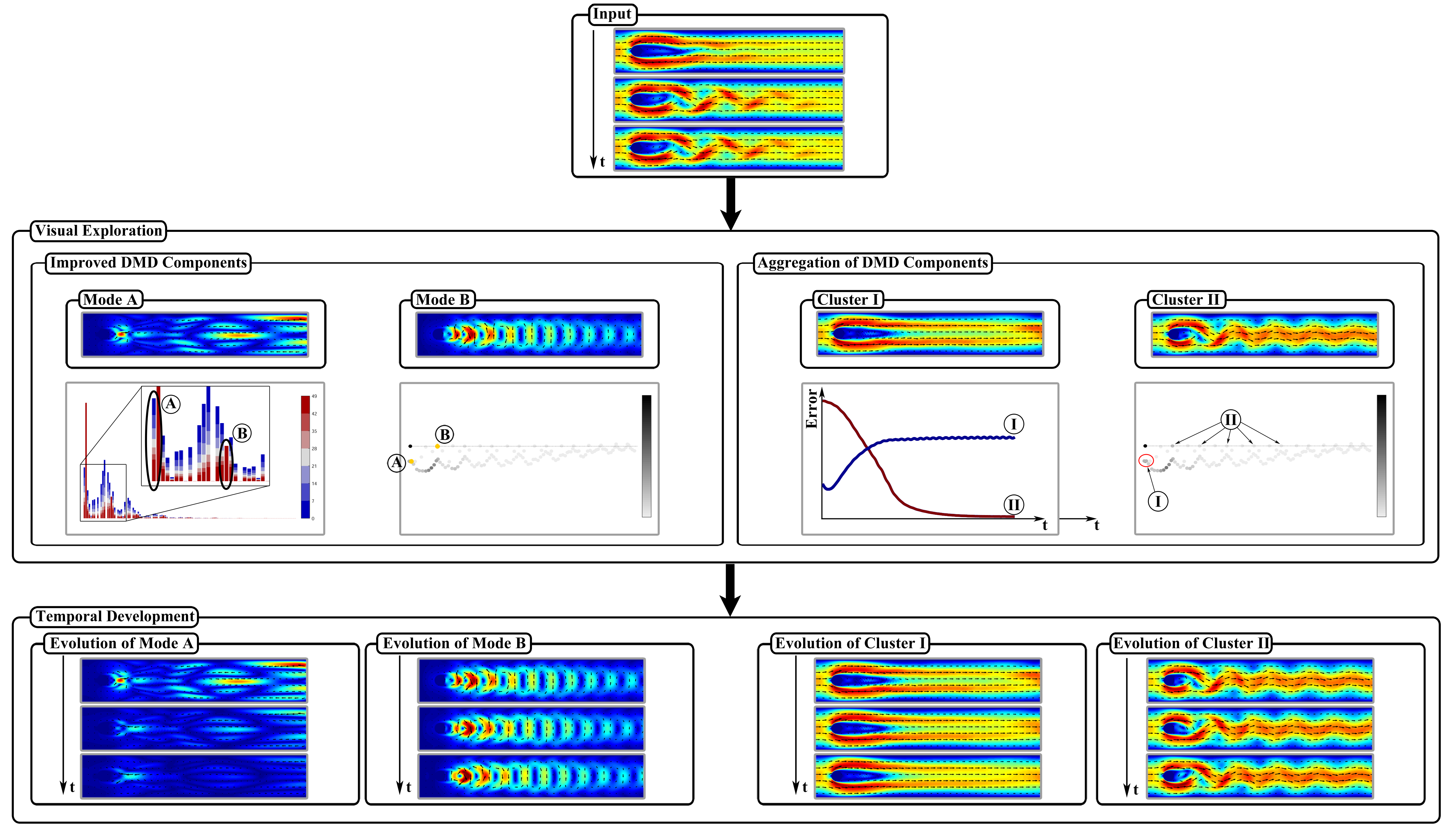}
	\caption{
	An overview of our proposed Dynamic Mode Decomposition (DMD) approach for the investigation of unsteady flow.
	For the data analysis process, we introduce improved DMD components and visualizations that respect the spatio-temporal character of DMD and, thus, describe the flow appropriately.
	Moreover, two new clustering approaches allow the user to aggregate flow components that segment the flow into physically relevant sections. 
	These can therefore be used for the selection of DMD components. 
	The combination of these new techniques enables new insights with DMD, as the resulting temporal developments suggest.
}
  \label{pic_overview}
}

\vgtcinsertpkg
\usepackage{amsmath,amssymb}

\usepackage{dsfont}
\usepackage{mathtools}
\usepackage{subfig}
\usepackage{algorithm}
\usepackage{algpseudocode}

\usepackage{ntheorem}

\theoremstyle{nonumberplain}
\theoremsymbol{}
\theoremseparator{}
\theoremnumbering{arabic}
\theoremheaderfont{\normalfont\bfseries}
\theorembodyfont{\normalfont\itshape}
\newtheorem{theorem}{Theorem}
\newtheorem{lemma}{Lemma}

\theoremstyle{nonumberplain}
\theoremseparator{:} 
\theoremheaderfont{\normalfont \bfseries}
\theorembodyfont{\normalfont} 
\theoremsymbol{\ensuremath{_\Box}} 
\newtheorem{proof}{Proof}

\begin{document}

\maketitle

\section{Introduction}
Dynamic Mode Decomposition (DMD) is a data-driven and model-free technique to decompose complex flows into fundamental spectral components.
These components correspond to spatio-temporal features that characterize periodicity, damping, (temporal) segmentation, and long-time behavior of the flow.
Basically, the algorithm results in three components: DMD modes, amplitudes, and eigenvalues.
Whereas the modes represent spatial contribution to the flow and the amplitudes specify their impact, each associated eigenvalue characterizes the temporal development.
The objective of this paper is to gain a better understanding of these components such that a more insightful analysis is achieved.
Moreover, as the visualization community has not paid much attention to DMD so far, we want to make DMD more accessible for both its users and the visualization research community.

DMD is supposed to identify spatial patterns associated with frequencies and growth rates that determine the behavior of a system.
So far, the investigation via DMD has been performed by the study of individual DMD components.
In addition, spatial and temporal properties of components are treated independently.
Since DMD is based on the interplay of spatio-temporal components, this traditional analysis process is insufficient.
It has several negative implications:
First, the relevance of the components to the entire system is not clearly specified.
Thus, an appropriate selection of components (for the analysis process) is not possible.
Second, the existing DMD visualizations could be misleading as the mutual dependencies of the components are not taken into account.

To address these problems, we focus on the representation of the components and their visualization. 
Our approach is guided by the needs of unsteady flow but could be extended to general time-dependent grid-based data.
Figure~\ref{pic_overview} illustrates the analysis process following our approach.
Our contributions can be summarized as follows:
\begin{itemize}
\item Conceptual contribution: We clarify drawbacks of the traditional DMD components and provide a new perspective on DMD based on a comparison with the discrete Fourier transform (DFT).
\item New visualizations: We improve DMD components and their representations using novel visualizations that respect the spatio-temporal character of DMD.
\item Data analysis contribution: We introduce two clustering approaches to aggregate components that segment the flow into physically relevant sections and can therefore be used for the selection of DMD components. 
\end{itemize}

We also discuss the mathematical foundation of DMD by providing a derivation and a specific formulation of DMD.
Moreover, with artificial and simulated examples, we show that our approach is able to identify characteristic features of unsteady flow fields.

\section{Related Work}
The visualization and analysis of unsteady flow are a challenging research topic.
A variety of decomposition techniques has been proposed to extract different kind of features from a flow.
The characteristics of a feature strongly depend on the method and are often difficult to define.
In the context of unsteady flow, we distinguish between two types of decomposition techniques.

The first type directly operates on the vector field, such as the Helmholtz Hodge decomposition (HHD) \cite{6365629} or the Morse decomposition \cite{Zhang2007}.
The HHD decomposes the vector field into two spatial components that are divergence- and curl-free, respectively.
Recent work~\cite{676} deals with an extension of the HHD based on Fourier transformation.
Wiebel et al.~\cite{4293009} propose a similar decomposition into a potential flow from the boundary and a localized flow to extract features. 
Rojo and G\"unter's~\cite{Rojo:2020VFTopoUnsteady} splitting method decomposes the flow into a steady and ambient part, enabling the description of the motion of topological elements and feature curves.
However, these decomposition techniques do not encode temporal patterns like DMD.
The Morse decomposition divides a vector field into disjoint invariant sets, called Morse sets. 
The connection of those Morse sets is illustrated by a directed graph giving an overview of the topological skeleton of the flow field.
While this approach highlights only spatial relations of the decomposition, we address spatio-temporal patterns. 
Bujack et al.'s state-of-the-art report \cite{Bujack:2020:PhysByMath} interprets physical features of several decomposition techniques of that type in terms of mathematical properties. 
We follow a similar approach, however, DMD is of another type and therefore extracts different features.

The second type of decompositions makes use of temporal coherence by performing the decomposition on the full time series, instead of considering each step individually.
Besides DMD, Principal Component Analysis (PCA) \cite{lumley67, berkooz:1993:pod_flows}, also kwown as Proper Orthogonal Decomposition (POD), is a technique of this type.
PCA hierarchically decomposes the data into an orthogonal basis of spatially correlated modes, called principal components (or POD modes), modulated by appropriate random time-coefficients. 
Therefore, the dynamic information is often neglected and particular emphasis is put on the spatial components.
Pobitzer et al's work~\cite{Pobitzer:2011:PODplus} deals with an extension of POD based on feature detectors.
The visualization of POD components and the additional use of feature detectors do not take spatio-temporal properties of the decomposition into account, which is again the main difference to our DMD approach.

DMD was introduced by Schmid and Sesterhenn \cite{schmid:2008:APS}. 
Schmid~\cite{schmid:2010:dmd_numerical_data} improved the DMD algorithm by using a reduced singular value decomposition (SVD).
Tu et al.~\cite{tu:2014:on_dmd_theorey_and_app} formulated the latest version of DMD, called exact DMD. 
As described in the introduction, the components of DMD are used separately and important spatio-temporal relations are not taken into account, especially not for the visualization.
In this way, DMD has been applied on diverse flow setups, e.g., the analysis of wake flows \cite{tu:2014:on_dmd_theorey_and_app, bagheri:2013:KoopamnCylinderWake, zhang:2014:idenfication_of_cs_using_pod_and_dmd, sampath:2014:pod_and_dmd_of_time_resolved_piv}, cavity flows \cite{seena:2011:dmd_of_turbulent_cavity_flows,lusseyran:2011:flow_coherent_structures}, mixing layer flow \cite{sayadi:2012:dmd_of_h_type, sayadi:2013:dmd_of_controlled_h_and_k_type}, and jet flows \cite{rowley:2009:spectral_nonlinear_flow, schmid:2011:applications_of_dmd}.
In the context of Lagrangian coherent structures, the interaction of DMD and Finite Time Lyapunov Exponent (FTLE) was considered to provide a feature-based description of the entire flow field \cite{ali:2017:turbulent_boundary_layer_lcs_pod_dmd, weheliye:2018:dmd_pod_lcs, Nair:2017:lcs_dmd}.
Kou and Zhang \cite{KOU:2017:ModeCriterion} propose a new criterion for the selection of dominant modes.
However, this approach uses traditional DMD components and is based on the energy over the full time.
Regarding visualization and computer graphics, DMD was used for background/foreground video separation \cite{kutz:2017:video_separation}, background modeling  \cite{pendergrass:2017:dmd_for_background_modeling, erichson:2016:compress_dmd_for_background_modeling}, edge detection \cite{bi:2017:dmd_edge_detection}, and the visualization of large-scale power systems \cite{mohapatra:2016:dmd_large_scale_power_systems}.
However, all above mentioned publications use the traditional DMD components.
With our novel visualizations of the improved components, we show that our techniques overcome drawbacks of the conventional DMD approach for visual flow analysis.

\begin{algorithm}[t]
	\caption{Exact Dynamic Mode Decomposition}\label{Algorithm_DMD}
    \begin{algorithmic}[1]
        \Function{DMD}{$x_0,\dots,x_m$}
		\State $X = \begin{bmatrix} x_0 & \dots & x_{m-1} \end{bmatrix}, Y = \begin{bmatrix} x_1 & \dots & x_{m}	\end{bmatrix}$
		\State Calculate the reduced SVD $X = U \Sigma V^*$ with $\textnormal{rank}(X) = r$.
		\State Calculate $S = U^* Y V \Sigma^{-1}$.
		\State Calculate $\lambda_1,\dots,\lambda_r$ and $v_1,\dots,v_r$ of $S$.
		\For {$\lambda_i \neq 0$}
		\State $\vartheta_i = \frac{1}{\lambda_i} Y V \Sigma^{-1} v_i$
		\EndFor
		\State $\Lambda = \text{diag}(\lambda_1,\lambda_2,\dots,\lambda_{r_0})$ with $\lambda_1, \lambda_2, \dots, \lambda_{r_0} \neq 0$
		\State $\Theta = \begin{bmatrix} \vartheta_1 &  \vartheta_2 & \dots & \vartheta_{r_0} \end{bmatrix}$
		\State Calculate $a = \Lambda^{-1} \Theta^+ x_1$ with $a = (a_1,\dots,a_{r_0})$
		\State Calculate if exist $c_0 = -\sum_{j=1}^{r_0} \frac{1}{\lambda_j} \prod_{k \neq j} \frac{1}{\lambda_j - \lambda_k}$
		\State\Return  $\lambda_j, a_j, \vartheta_j, c_0$
		\EndFunction
	\end{algorithmic}
\end{algorithm}

\begin{figure*}[t]
	\centering
	\includegraphics[width=\linewidth]{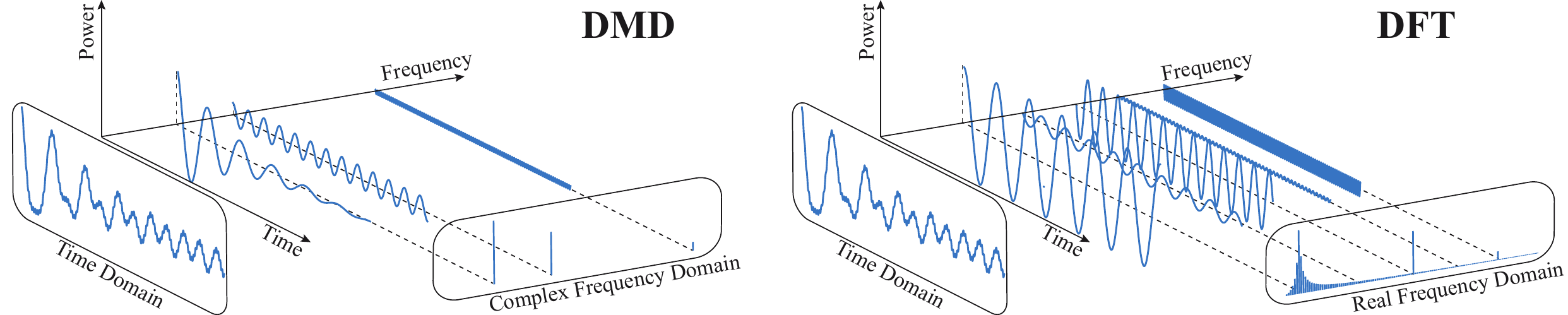}
	\caption{
		A comparison of DMD (left) with DFT (right). 
		Both methods lead to a structurally similar decomposition.
		However, DMD computes its own complex frequencies (depending on a frequency and growth rate) by the eigenvalues of a least-squares fit matrix, whereas DFT uses fixed real frequencies given by the (complex-valued) roots of unity.
		The components of DMD can therefore converge or diverge over time, which creates new opportunities for investigating data as the simple complex frequency domain highlights.
		} 
	\label{pic_DMD_DFT}
\end{figure*}

\section{Mathematical Foundation}
We concentrate on time-dependent flow fields defined on a grid and sampled uniformly in time.
More precisely, for a 2D velocity field on a grid with $N$ points at time $t_k$, the data is characterized by $u_1(t_k), v_1(t_k), \dots, u_N(t_k), v_N(t_k) \in \mathbb{R}$, where $u$ and $v$ represent the velocity components.
For DMD, the snapshots $x_0,\dots,x_m \in \mathbb{R}^{2N}$ are given by
\begin{equation*}
\begin{bmatrix}
|	& |		& 		& |			\\ 
x_0 &  x_1 	& \dots & x_{m}	\\
|	& |		& 		& |			\\ 
\end{bmatrix}
=
\begin{bmatrix}
u_1(t_0)	& 	u_1(t_1)		& 		& u_1(t_m)	\\ 
v_1(t_0) 	& 	v_1(t_1) 	& \dots 	& v_1(t_m)	\\
\vdots & \vdots & \ddots & \vdots \\
u_N(t_0)	&	u_N(t_1)		& 		& u_N(t_m)	\\ 
v_N(t_0) 	&   v_N(t_1) 	& \dots 	& v_N(t_m)	\\
\end{bmatrix}.
\end{equation*}
A 3D flow field is defined analogously with three instead of two components.

Before formulating the algorithm of DMD, we summarize the concepts of DMD \cite{kutz:2016:DMD_Book,krake:2019:DMD} and compare them with DFT.
These aspects help understand the visualization techniques and their principles.

\subsection{Derivation of DMD}
DMD performs a spectral decomposition on arbitrary data. 
Therefore, we generally consider complex-valued snapshots $x_0, x_1, \dots, x_m \in \mathbb{C}^{n}$, which are usually uniformly sampled in time.
The size of a snapshot is typically substantially greater than the number of snapshots, i.e., $n \gg m$.
For the snapshots, we consider the following minimization problem:
\begin{equation} \label{eq:basic_minimization_problem}
\min_{A \in \mathbb{C}^{n \times n}}  \sum_{j=0}^{m-1} \lVert A x_j - x_{j+1} \rVert_2^2~,
\end{equation} 
where $\lVert \cdot \rVert_2$ denotes the Euclidean norm. 
In other words, we search a matrix $A$ that optimally connects the subsequent snapshots in a least squares sense.
The idea of DMD is to calculate a low-dimensional representation of $A$ and to perform an eigenvalue decomposition on it to detect frequency patterns in the data.
To this end, we insert the snapshots as column vectors into the following two matrices:
\begin{equation} \label{eq:data_matrix}
X = \begin{bmatrix}
|	& 		& |			\\ 
x_0 & \dots & x_{m-1}	\\
|	& 		& | 
\end{bmatrix},
\quad
Y = \begin{bmatrix}
|	& 		& |			\\ 
x_1 & \dots & x_{m}		\\
|	& 		& |
\end{bmatrix}.
\end{equation} 
The optimization problem in Equation~\ref{eq:basic_minimization_problem} can now be rewritten as $\min_{A \in \mathbb{C}^{n \times n}}  \lVert A X - Y \rVert_F^2~$, where $\lVert M \rVert_F \coloneqq (\sum_{i = 1}^n \sum_{j = 1}^m \lvert m_{ij} \rvert)^{1/2}$ denotes the Frobenius norm. 
Consequently, a best-fit matrix is explicitly given by
\begin{equation}
A = YX^+ \in \mathbb{C}^{n \times n}~,
\end{equation}
where $X^+$ is the Moore-Penrose pseudoinverse, a generalized inverse \cite{moore:1920:moore_penrose} of the non-square matrix $X$.
Note that $A$ is a (large) matrix of size ${n \times n}$.
We assume that $A$ is diagonalizable (which is almost always the case, if $n \gg m$).
This means that there exists an invertible matrix $V = \begin{bmatrix} v_1 & v_2 & \dots v_n \end{bmatrix} \in \mathbb{C}^{n \times n}$ consisting of eigenvectors, as well as a diagonal matrix $\Lambda = \text{diag}(\lambda_1,\lambda_2,\dots,\lambda_n) \in \mathbb{C}^{n \times n}$ with corresponding eigenvalues, such that $V^{-1} A V = \Lambda$.
Therefore, we can approximate the $k$th snapshot by
\begin{equation} \label{eq:basic_decomposition}
x_k \approx A^k x_0 
= V \Lambda^k V^{-1} x_0 
= \sum_{j=1}^n b_j \lambda_j^k v_j~, 
\end{equation}
where $b = \begin{pmatrix} b_1 & b_2 & \dots & b_n \end{pmatrix}^T = V^{-1} x_0$. 
Note that entries of $b$ are coefficients of the linear combination of $x_0$ in the eigenvector basis.

Given that $n \geq m$, the rank of $A$ cannot be higher than $m$ due to the dimension restriction and, hence, the number of non-zero eigenvalues is at most $m$.
Consequently, the dynamic behavior will be captured by $m$ summands, consisting of the triple $(b_j,\lambda_j,v_j) \in \mathbb{C} \times \mathbb{C} \times \mathbb{C}^n$ for $j=1,2,\dots,m$.
DMD calculates exactly these (non-zero) eigenvalues and eigenvectors referred to as DMD eigenvalues and DMD modes, respectively.
For simplicity, we denote them as eigenvalues and modes throughout the rest of the paper.
The coefficients $b_1,\dots,b_m$ in the decomposition of Equation~\ref{eq:basic_decomposition} need to be computed by DMD in a different way.
We call them DMD amplitudes.
However, as the modes usually do not form a basis, an error will occur in the reconstruction of the first snapshot $x_0$. 

\subsection{DMD Algorithm} \label{ssec:DMD_algorithm}
Algorithm~\ref{Algorithm_DMD} shows the DMD method derived from the previous section. 
However, it differs in the calculation of amplitudes from the standard literature. 
The new definition of amplitudes, which uses the second snapshot instead of the first one for the reconstruction, was introduced by Krake et al. \cite{krake:2019:DMD}.
We make use of this new formulation by improving the representation of DMD components and integrating them to arrive at novel and more adequate DMD visualizations.

In Algorithm~\ref{Algorithm_DMD}, a triple $(a_j,\lambda_j,\vartheta_j)$ similar to the one in Equation~\ref{eq:basic_decomposition} is determined without explicitly computing $A$. 
This is achieved by using a reduced SVD applied to the data-matrix $X$, which results in two unitary matrices $U \in \mathbb{C}^{n \times r}$ and $V \in \mathbb{C}^{m \times r}$ (which are real-valued for real-valued data) as well as a real-valued diagonal matrix $\Sigma \in \mathbb{R}^{r \times r}$  with $r = \text{rank}(X)$ (lines 2--3).
The low-dimensional representation of $A$ is given by (line 4)
\begin{equation} \label{eq:matrixS}
S = U^* A U = U^* Y V \Sigma^{-1}~.  
\end{equation}
Next, we compute the eigenvalues $\lambda_j$ and eigenvectors $v_j$ of $S$ for $j = 1,\dots,r$ (line 5).
Then, we transform the eigenvectors with non-zero eigenvalues into modes $\vartheta_j$ (lines 6--8). 
Finally, the amplitudes $a_j$ and the error scaling $c_0$ are calculated (lines 9--12).
By this choice of amplitudes, Krake et al. \cite{krake:2019:DMD} proved the following reconstruction property:
\begin{equation} \label{eq:theroem} 
x_0 = \sum_{j=1}^m {a_j \vartheta_j} + c_0 \cdot q~,
\qquad
x_k = \sum_{j=1}^m {\lambda_j^k a_j \vartheta_j}~,
\end{equation}
for $k = 1,\dotso,m$ and $q = x_m - XX^+ x_m$.
This property holds when both $x_0,\dotso,x_{m-1}$ and $x_1,\dots,x_m$ are linearly independent and $\lambda_1,\dots,\lambda_m$ are distinct (in this case $r_0 = r = m$).
The assertion states that DMD inherits the property of Equation~\ref{eq:basic_decomposition} with appropriate coefficients.
This version of DMD captures the entire system by providing a structured spectral decomposition into temporal and spatial aspects.
Thus, we are able to clarify the impact of (aggregated) components precisely.
Moreover, the interplay of the eigenvalues, modes, and amplitudes can be interpreted in a new and clearer way.
Based on these insights, we create novel visualizations that respect the spatio-temporal character of DMD.

\subsection{DMD and DFT} \label{ssec:DMD_DFT}
Discrete Fourier transform (DFT) is a well-understood tool to analyze time-dependent data.
It is able to extract frequency-based features like periodicity. 
Since DFT and DMD lead to a structurally similar decomposition of data, it is possible to translate properties and procedures from DFT to DMD.
More precisely, DFT yields 
\begin{equation} \label{eq:DFT}
x_k = \sum_{j = 0}^m \mu_j^k  \hat{x}_j~,
\end{equation}
where $\mu_j = e^{\frac{2 \pi j}{m+1}}$ are the roots of unity and $\hat{x}$ is the Fourier transform.
Thus, the DFT converts 1D snapshots $x_k$ into complex numbers $\hat{x}_j$ (for vector-valued data $x_k$ into complex vectors $\hat{x}_j$) each depending on a root of unity $\mu_j$.
Using the exponential form, i.e., $\mu_j = e^{i \varphi_j}$, we observe that the components $\mu_j$ have magnitude $1$.
Thus, they are only determined by a real frequency $f_j = \frac{\varphi_j}{2\pi}$.
This leads to a real frequency domain representation as shown in Figure~\ref{pic_DMD_DFT} on the right side, where the magnitude of a summand $\mu_j^k \hat{x}_j$ does not change over time.

In contrast, DMD computes the triplets: eigenvalues $\lambda_j$, modes $\vartheta_j$, and amplitudes $a_j$. 
A direct comparison of the components is difficult since the interplay of DMD components is more complex than for DFT.
Nonetheless, we can bring DMD and DFT in accordance using Equations \ref{eq:theroem} and \ref{eq:DFT}.
In the decompositions, the temporal components are given by $\lambda_j$ and $\mu_j$, respectively.
Each set of temporal components characterize a decomposition entirely as the spatial contributions, i.e., the Fourier transformed vectors $\hat{x}_j$ or the scaled modes $a_j \vartheta_j$, are simply fitted to those in a unique way.
Therefore, the decompositions only differ in the choice of temporal components.
In fact, the eigenvalues $\lambda_j$ computed by DMD are variable, whereas the DFT uses roots of unity $\mu_j$. 
Hence, the eigenvalues $\lambda_j = r_j e^{i \varphi_j}$ are characterized by both a frequency $f_j = \frac{\varphi_j}{2\pi}$ and a magnitude $r_j$, or, in other words, by a complex frequency.
Thus, DMD produces a complex frequency representation of the data as illustrated in Figure~\ref{pic_DMD_DFT} on the left side.

The DMD procedure can be summarized as a two-stage method:
First, DMD computes appropriate complex frequencies based on the data.
Then, the data is transformed into complex numbers (for vector-valued data into complex vectors) that depend on those complex frequencies.
From a certain point of view, DMD can thus be seen as an extension of DFT.  
This opens up new possibilities to extract complex frequency-based features, like periodicity, damping, and temporal segmentation.
In Figure~\ref{pic_DMD_DFT}, we observe that DMD needs three (non-vanishing) complex frequencies to reconstruct the signal, whereas DFT a number of real frequencies.
In a more complex scenario, the challenge is to find relevant DMD components. 
This task is more complicated than in the case of DFT, since the interplay of eigenvalues $\lambda_j$, modes $\vartheta_j$, and amplitudes $a_j$ needs to be taken into account in order to respect the spatio-temporal character of DMD.

Starting with the eigenvalues $\lambda_j$ (complex frequencies), these should be used to highlight the impact of modes over time.
Previously, the modes $\vartheta$ should be adjusted by their amplitudes $a_j$ according to DFT such that the scaled modes $a_j \vartheta_j$ (analogously to the Fourier-transformed vectors) characterize the spatial contribution and the impact.
Based on these observations, we provide improved components and new visualizations clarifying the interplay of components.
Moreover, a cluster method is developed on the principles of DFT (aggregation of harmonics) that is used for the selection of relevant DMD components.

\begin{figure}[t]
	\centering 
	\includegraphics[width=\columnwidth]{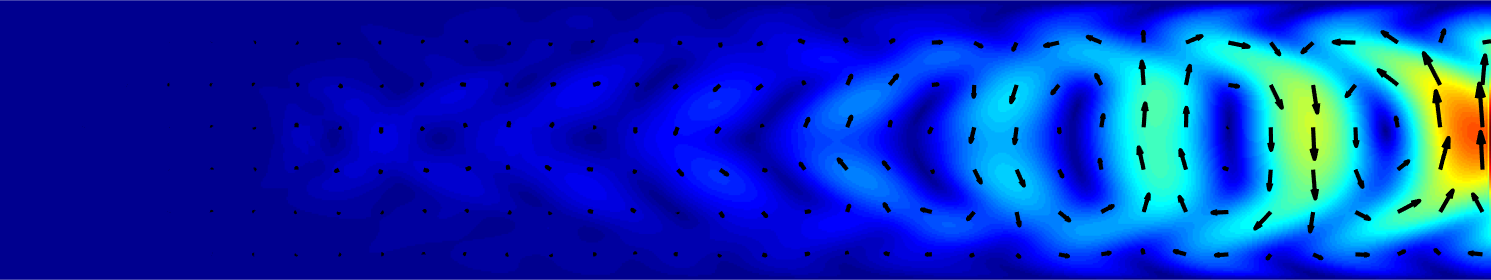}
	\caption{
	Traditional visualization of a 2D mode $\vartheta_j$ with arrow glyphs and color-coded velocity magnitude (blue = low, red = high).
	}
	\label{pic_karman_modes_old}
\end{figure}

\section{Investigation Approach} \label{sec:4:VisTech} 
In this section, we first describe the conventional DMD approach with the traditional DMD components and visualizations by highlighting their benefits and drawbacks.
Then, we show how these components can be improved and visualized in a novel way to resolve these issues.
On this basis, we present two clustering methods that segment the flow into physically relevant sections.
Finally, an approach for the selection of DMD components is proposed.

\begin{figure}[t]
	\centering
    {\includegraphics[width=0.63\columnwidth]{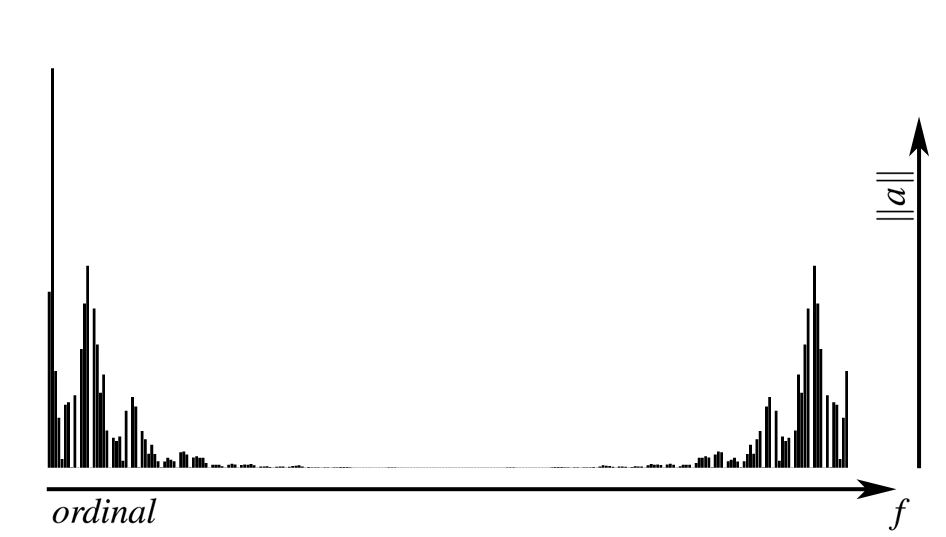}}
	\hfill
    {\includegraphics[width=0.36\columnwidth]{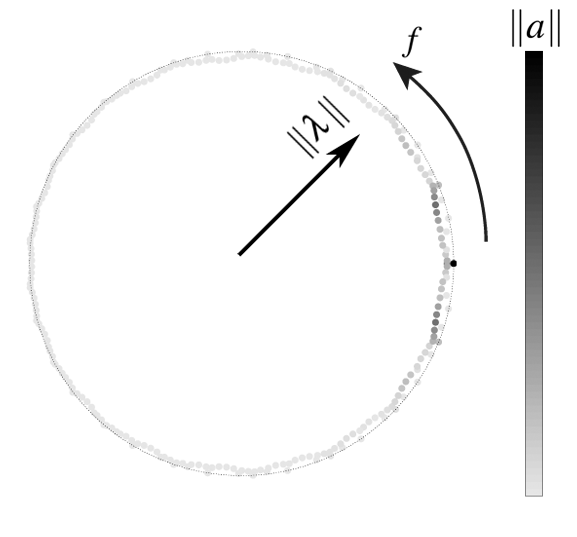}}
	\caption{
	The traditional visualization of the amplitudes (left) and eigenvalues (right) for the 2D Karman vortex street.
	The absolute values of amplitudes $\lvert a_j \rvert$ are visualized in a bar diagram that is sorted by the frequency of their  corresponding eigenvalues. 
	The visualization neither takes temporal aspects nor connections to the eigenvalue visualization into account.
	The eigenvalues are represented in the complex plane encoded by their corresponding amplitudes, where light gray indicates insignificant components.
	For real-valued data (which is the case for flow data), the eigenvalues occur in complex conjugate pairs.
	Hence, redundant information is visualized.
	Moreover, the radial representation has several drawbacks for the identification of patterns.
	}
	\label{pic_karman_old_ew_dom}
\end{figure}

\subsection{Traditional DMD Components}
In the following, we point out the pros and cons of the conventional DMD approach. 
DMD computes the following triplets: modes $\vartheta_j$, amplitudes $a_j$, and eigenvalues $\lambda_j$. 
Traditionally, the three groups of components are mainly considered and visualized separately.

\paragraph{\textbf{DMD Modes}}
The modes represent the spatial contribution to the flow field and are supposed to display local and global features, like symmetry, mixing, transient response, and long-time behavior.
Each entry of a vector-valued mode $\vartheta_j$ corresponds to a spatial location of the velocity field as shown in Figure~\ref{pic_karman_modes_old}.
For the investigation with DMD, traditionally, the modes are considered individually.
However, as DMD is based on the superposition principle, modes that are inspected must be selected carefully. 
Otherwise, the validity of spatial features is not ensured. 
Moreover, the modes are complex-valued and therefore mostly visualized by their real and imaginary part such that an interpretation is even more complicated.
A noteworthy mode is the one with corresponding eigenvalue $\lambda_j = 1$. 
If it exists, it represents the constant flow, often called time-averaged flow, which does not change over time.

\paragraph{\textbf{DMD Amplitudes}}
The amplitude $a_j$ determines the influence of the mode $\vartheta_j$ to the flow.
As mentioned before, the amplitudes are traditionally computed by $a=\Theta^+x_0$ instead of $a=\Theta^+x_1$ (compare Algorithm~\ref{Algorithm_DMD}) such that a reconstruction like in Equation~\ref{eq:theroem} does not hold.
These complex numbers are then visualized by their absolute value, i.e., $\lvert a_j \rvert$, in a bar diagram sorted by the frequency of their corresponding eigenvalues as illustrated in Figure~\ref{pic_karman_old_ew_dom} on the left.
The visualization gives an overview of the distribution of modes' influence.
On that basis, the selection of relevant modes is performed.
More precisely, a choice of the $k$-th most dominant modes is made by the $k$ highest absolute values of amplitudes.
This approach does not take temporal patterns into account, since the amplitudes only precisely reflect the influence of modes at time $t_0$.
Additionally, the influence is not precisely exact for the traditional amplitudes and the corresponding modes are not necessarily normalized.
Therefore, a misleading choice of modes could be made that do not represent any features of the flow, especially the ones that evolve over time.
Moreover, redundant components are highlighted in Figure~\ref{pic_karman_old_ew_dom} (as we will prove later) that may lead to unnecessary selection and visualization of modes.

\begin{figure}[t]
	\centering
	\includegraphics[width=\columnwidth]{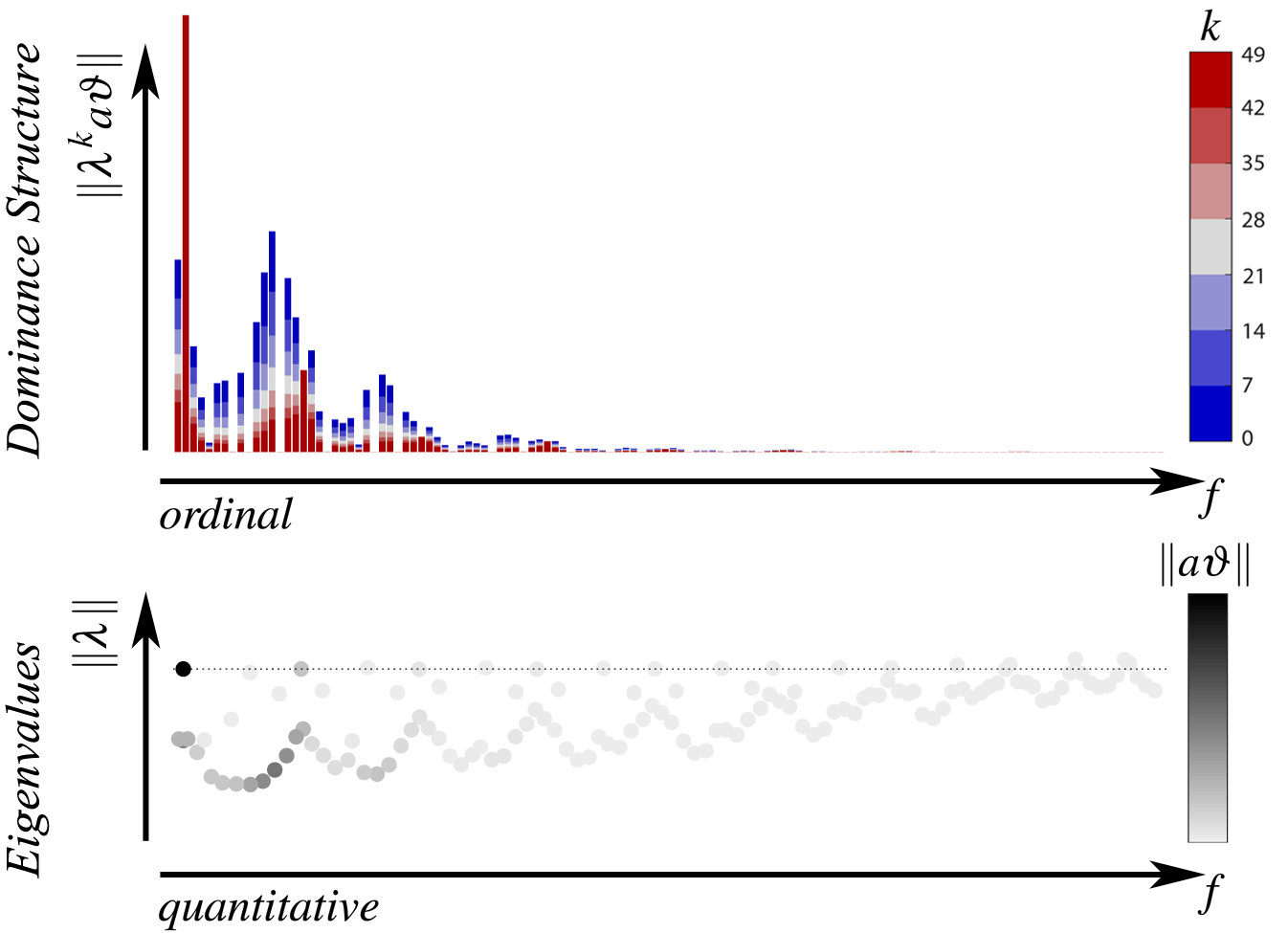}
	\caption{Our proposed visualizations of the eigenvalues (bottom) and the dominance structure (top) for the 2D Karman vortex street.
	The non-redundant eigenvalues are represented by their frequencies (arguments) and growth rates (magnitudes) and are grayscaled by the norms of their associated scaled modes, where light gray indicates insignificant components.
	The dominance structure visualizes the influence of non-redundant scaled modes.
	For this, the norms of non-redundant scaled modes are illustrated in a bar diagram that is sorted by the frequencies of their corresponding eigenvalues.
	In addition, the temporal development of each scaled mode at specific time steps is integrated into the representation.
		}
	\label{pic_karman_new_ew_dom}
\end{figure}

\paragraph{\textbf{DMD Eigenvalues}}
An eigenvalue of a corresponding mode describes its temporal behavior. 
The position within the complex plane provides information about frequency and growth rate. These quantities are given by their arguments and magnitudes, respectively.
Usually, the eigenvalues are (gray) scaled by the absolute value of their corresponding amplitudes, as illustrated in Figure~\ref{pic_karman_old_ew_dom} on the right.
For real-valued data, the eigenvalues occur in complex conjugate pairs, i.e., if $\lambda = r \; e^{i\varphi}$ is an eigenvalue, then $\overline{\lambda} = r \; e^{-i\varphi}$ is an eigenvalue, too. 
Therefore, the traditional visualization in the complex plane, as shown in Figure~\ref{pic_karman_old_ew_dom} on the right, displays redundant information.

Using the exponential form $\lambda = r \; e^{i\varphi}$ with frequency $f = \frac{\varphi}{2\pi}$ and magnitude $r$ as well as Equation~\ref{eq:theroem}, the eigenvalues can be categorized in the following way:
\begin{itemize}
	\item $r_j<1$: These components describe transient responses, because the potentiation of the eigenvalues will cause the component to vanish.
	\item $r_j>1$: Such components show contrary behavior: due to the potentiation of the eigenvalue, the components will grow and diverge.
	\item $r_j=1$: In this case, potentiation of the eigenvalue will cause a rotation on the unit circle with a specific frequency. 
	These components characterize the steady state.
\end{itemize}
Even though the categorization into these three cases is possible and the interpretation of eigenvalue potentiation is more accessible (see Figure~\ref{pic_karman_old_ew_dom}), the detection of recurring and salient patterns suffer from the radial representation.

\subsection{Improved DMD Components} 
In this subsection, we present an improvement to the conventional DMD components and visualizations that resolves the issues explained in the previous subsection.

\paragraph{\textbf{DMD Eigenvalues}}
For real-valued data (as given in many applications), the eigenvalues occur in complex conjugate pairs (compare Figure~\ref{pic_karman_old_ew_dom} on the left).
This well known fact follows from a real-valued SVD resulting in the real-valued matrix $S$ (see Equation~\ref{eq:matrixS} or appendix).
Therefore, the representation of the eigenvalues can be restricted to the upper half plane (complex numbers with non-negative imaginary part), which hides redundant components and leads to a clearer view.
Furthermore, the visualization in the complex plane has some disadvantages as well.
It is difficult to detect recurring and salient patterns due to the radial representation of the eigenvalues.
Therefore, we propose a characterization with more emphasis on the frequency and growth rate.
To this end, we remove the redundant eigenvalues and plot the remaining ones in a coordinate system with axes set by argument and magnitude, as visualized in Figure~\ref{pic_karman_new_ew_dom} at the bottom.
This representation enhances the identification of repetitive patterns significantly, which is the basis for the cluster methods later on.
In addition, the eigenvalues are grayscaled by the norms of the so-called scaled DMD modes that will be presented in the following.
In short, the norms represent the correct influence of components.

\begin{figure}[t]
	\centering
    \includegraphics[width=\columnwidth]{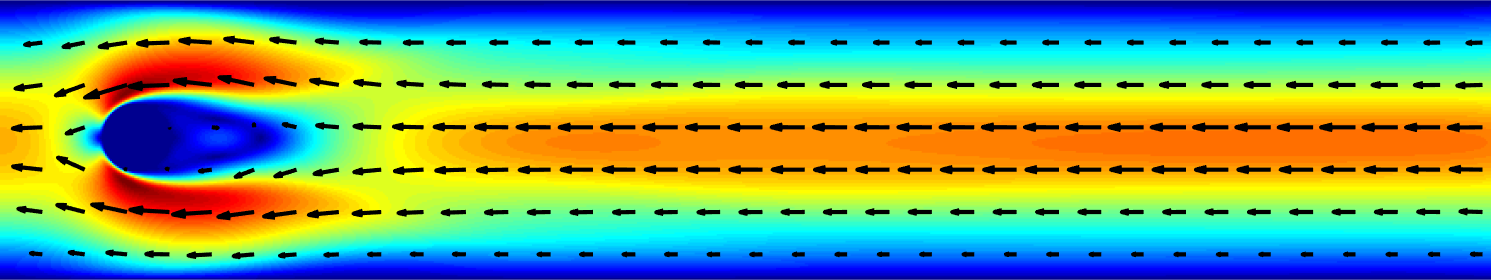}
	\hfill
    \includegraphics[width=\columnwidth]{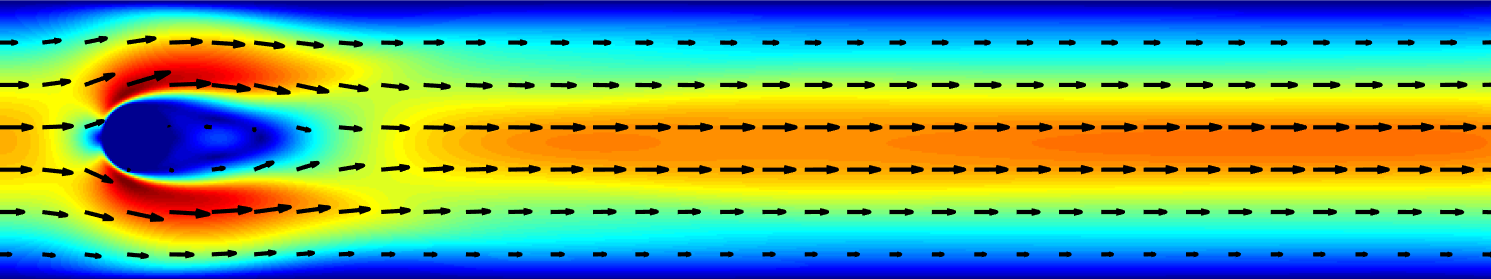}
	\caption{
	The traditional DMD mode to the eigenvalue $\lambda = 1$ is compared to our scaled DMD mode.
	Both are taken from the 2D Karman vortex street data set.
	The velocity vectors of the traditional DMD mode point in the wrong direction, as the inflow is on the left side.
	Scaling the DMD mode resolves this issue. 
	}
	\label{pic_dmd_mode_EW1_compare}
\end{figure}

\paragraph{\textbf{Scaled DMD Modes}}
The DFT can be linked to DMD because DMD can be seen as a two-stage method where eigenvalues computed first and the modes and amplitudes are fitted subsequently to those.
The Fourier-transformed vectors are the counterpart to the modes $\vartheta_j$ multiplied with their amplitudes $a_j$ (see Section~\ref{ssec:DMD_DFT}).
Thus, we propose scaling the modes for the analysis and denoting these new objects $a_j \vartheta_j$ as scaled DMD modes.
Due to this combination, the new representation contains both the spatial contribution to the decomposition and the influence to the system.
Before we discuss these two aspects, some advantageous mathematical properties of the representation are presented.

One important property of the scaled modes is that they occur in complex conjugate pairs, i.e., the scaled modes of a complex conjugate pair of eigenvalues $\lambda, \overline{\lambda}$ are given by $a \vartheta, \overline{a \vartheta}$.
A detailed mathematical proof of this fact can be found in the appendix.
Hence, the superposition of these two scaled modes can be expressed by twice the real part, since $a  \vartheta + \overline{a  \vartheta} = 2 \Re(a  \vartheta)$. 
Furthermore, the spatio-temporal development is given by 
\begin{equation} \label{eq:combining_counterparts}
\lambda^k a \vartheta + \overline{\lambda^k a  \vartheta}
= 2 \Re(\lambda^k a  \vartheta)~, \quad k \in \mathbb{N}~.
\end{equation}
As a result, we can combine these complex conjugate pairs, which facilitates the analysis approach.
In particular, the number of eigenvalues $\lambda_j$ and scaled modes $a_j \vartheta_j$ (with an eigenvalue having a non-zero imaginary part) is reduced by a factor of two.
This aspect is crucial for all following steps.

\begin{figure*}[t]
	\centering
	\includegraphics[width=\textwidth]{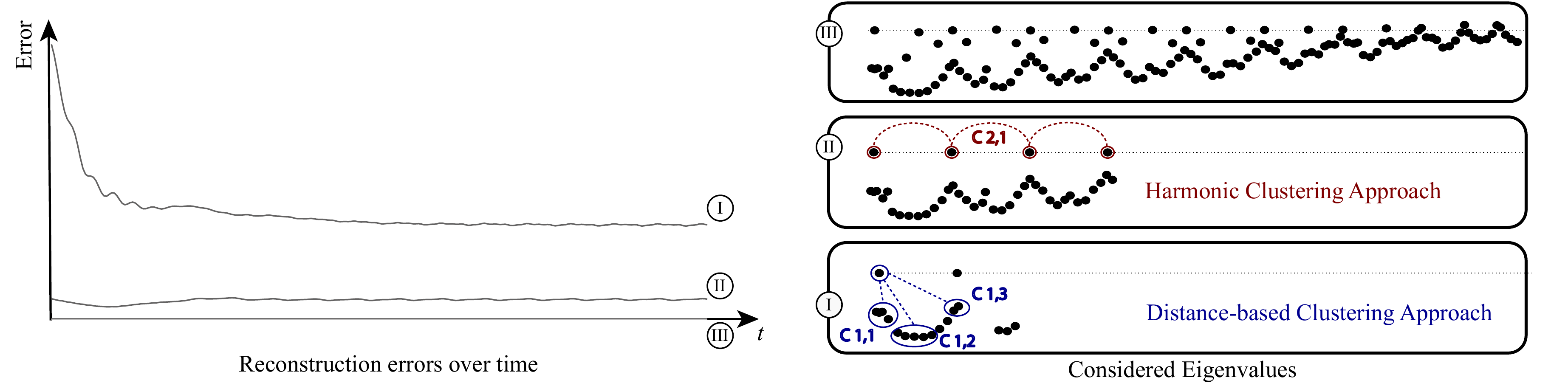}
	\caption{
	The reconstruction error over time (left) of different subsets of components (I, II, and III) from the 2D Karman vortex street data set.
	The eigenvalues of the respective clusters are shown on the right.
	Whereas the subset III consists of all components, leading to an error-free reconstruction of the flow, the subsets I and II are thinned out and used for the cluster approaches that lead to a selection of modes.
	Different thresholds facilitate the identification of patterns and support the clustering of components (red and blue on the right).
	}
	\label{pic_karman_reconst_and_cluster}
\end{figure*}

\textit{(a) Spatial Properties:}
The multiplication of a mode $\vartheta_j$ by its (complex) amplitude $a_j$ corrects its orientation. 
To demonstrate this, the traditional mode to the eigenvalue $\lambda = 1$ is depicted in Figure~\ref{pic_dmd_mode_EW1_compare} (top), often referred to as time-averaged flow.
We observe that the velocity vectors point in the wrong direction, as the flow moves from the left to the right.
Since this mode is real-valued and the only one that has impact on the boundary regions, the flow direction is truly incorrect and therefore the mode does not display the physical phenomena correctly (the time-averaged flow).
This is due to the fact that the modes are eigenvectors, which can be scaled by any complex factor.
To represent real physical properties, a mode $\vartheta_j$ needs to be linked to the data by multiplying it with its amplitude $a_j$.
In Figure~\ref{pic_dmd_mode_EW1_compare} (bottom), the scaled mode to the eigenvalue $\lambda = 1$ is visualized.
Now, the velocities point in the correct direction.
This simple example demonstrates that the combination of amplitudes and modes is crucial for the interpretation of spatial features and improves the fundamental (numerical) representation.

For the visualization of scaled modes, we make use of Equation~\ref{eq:combining_counterparts}.
In the analysis of real-valued data, the scaled modes and, in particular, their temporal development can thus be restricted to the real part, since the imaginary part will vanish in a superposition, as Equation~\ref{eq:combining_counterparts} shows.

\textit{(b) Influence Properties:}
The traditional amplitudes (typically used for the selection of modes) suffer from an imprecise computation as Equation~\ref{eq:theroem} does not hold exactly.
As a result, the influence of modes is not reflected correctly.   
Using the proposed improved representation, i.e., the scaled mode $a_j \vartheta_j$ computed by Algorithm~\ref{Algorithm_DMD}, the influence is given by its norm $\lVert a \vartheta \rVert$.
This formula is in accordance to the influence of Fourier-transformed vectors and is more precise.

Even though a correct choice of the most dominant components is now more likely, the norm only represents the influence at time $t_0$.
A selection of components only based on the norm of scaled modes (or simply on the absolute value of amplitudes) is insufficient as a temporal encoding is not included (compare Figure~\ref{pic_karman_old_ew_dom} on the left).
To integrate this, we propose a visual representation, referred to as dominance structure, that uses both the scaled modes and the eigenvalues.
This representation is illustrated in Figure~\ref{pic_karman_new_ew_dom} (top).
Basically, the norms of the scaled DMD modes $\lVert a \vartheta \rVert$ are visualized in a bar diagram that is sorted by the frequency of their corresponding eigenvalues.
However, it also takes the following aspects into account:
\begin{itemize}
    \item Since the scaled modes occur in complex conjugate pairs (like the eigenvalues), it is sufficient to display one of them, as they have the exact same impact.
    \item Due to the sorting according to the frequency (argument), it is possible to differentiate between the impact of low and high frequencies, which is in the spirit of DFT (see Figure~\ref{pic_DMD_DFT}) and the eigenvalue visualization (see Figure~\ref{pic_karman_new_ew_dom}).
    Therefore, the two visualizations can be linked as both rely on the same quantities. 
	More precisely, the frequency is represented in an ordinal way for the dominance structure (i.e., sequence of numbers) and a quantitative way for the eigenvalues (i.e., exact positions).
    \item The temporal development of a scaled mode $a_j \vartheta_j$ (according to the norm) is integrated into the representation.
    This development is given by the values $\lVert \lambda_j^k a_j \vartheta_j \rVert$ at specific time steps $k$, which are visualized by color-coded bars.
    The corresponding time steps $k$ are illustrated by a color map on the right.
    For the sake of visibility, a color-coded bar should be always either placed in the foreground (if $\lvert \lambda_j \rvert<1$) or background (if $\lvert \lambda_j \rvert>1$).
\end{itemize}
Since the time step $k=0$ (here, highlighted by dark blue) is considered, the visualization is an extension and an improvement to the traditional one shown in Figure~\ref{pic_karman_old_ew_dom}. 
It allows for new insights into the influence of scaled modes and supports the understanding of the interplay.

\begin{figure*}[t]
	\centering
	\includegraphics[width=\linewidth]{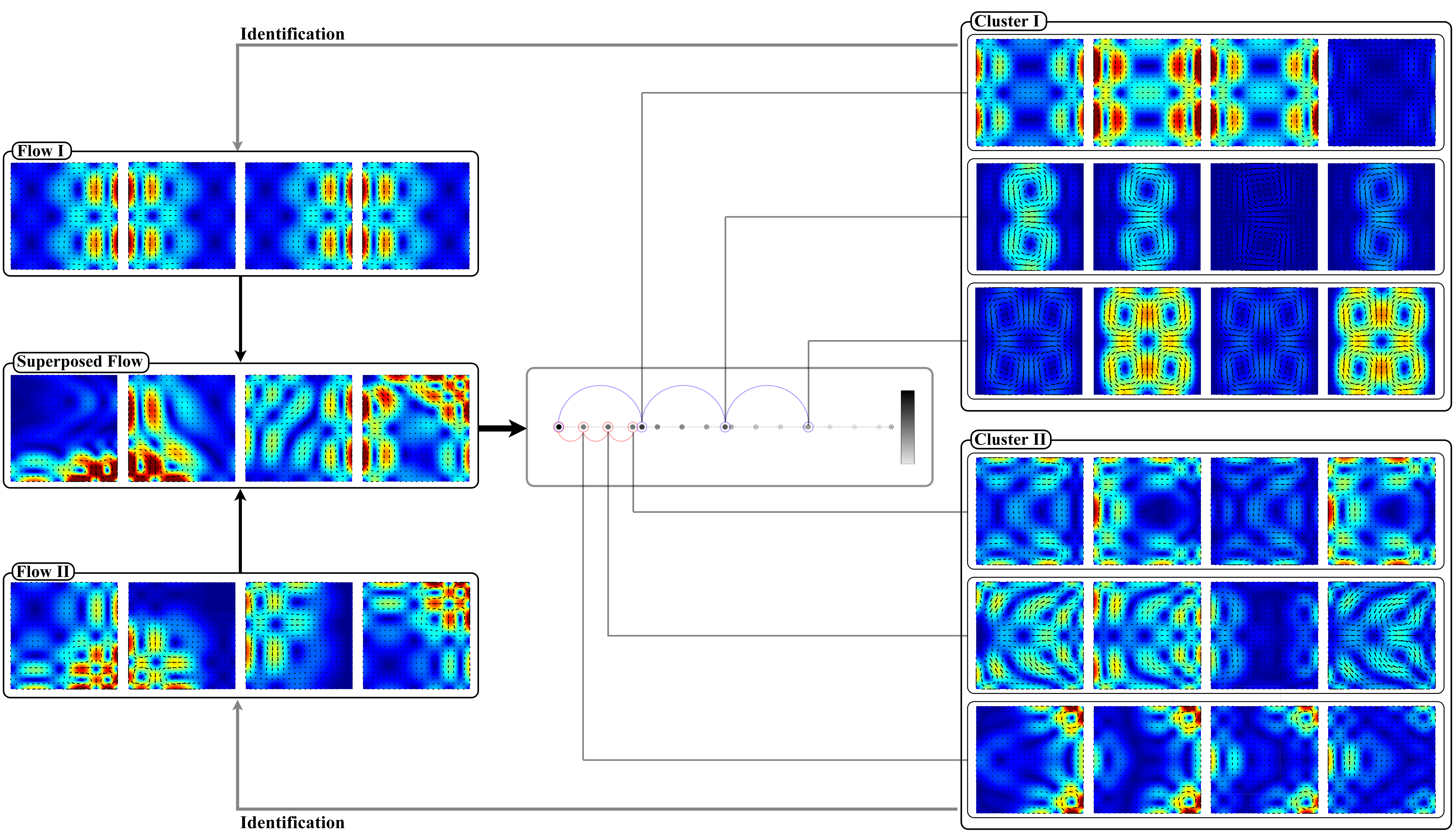}
	\caption{
	The procedure of the harmonic cluster approach applied to the superposed quadgyre data set:
	On the left, selected snapshots from the two unsteady flow fields (Flow I and II) are shown that characterize the superposed quadgyre.
	These data sets consist in each case of four vortices moving periodically from left to right (Flow I) or from top right to the bottom right in a crescent-shaped move (Flow II).
	The eigenvalue representation in the middle reveals two different frequency patterns that can be captured by the harmonic clustering approach.
	For each aggregation, the three most dominant modes are represented by their temporal development. 
	It can be observed that each aggregation characterizes one respective base flow, which verifies the usefulness of the cluster method.
	}
	\label{pic_superposed_quadgyre}
\end{figure*}

\subsection{Aggregation and Selection of Components}
DMD decouples time-dependent flow into spatial and temporal components.
For an appropriate selection of components, however, the spatio-temporal character of DMD needs to be taken into account.
Mathematically, a selection of components can be formalized by a subset $\mathcal{C} \subseteq \{1,\dots,r_0\}$ of all components.
Having defined a subset $\mathcal{C}$, we denote the temporal development of it at step $k$ as:
\begin{equation} \label{eq:cluster}
	\lambda^k_\mathcal{C}a_\mathcal{C}\vartheta_\mathcal{C} \coloneqq \sum_{i\in\mathcal{C}}{\lambda^k_i a_i \vartheta_i}~.
\end{equation}

If the traditional dominance-based approach for the selection of components is consulted, then the components are chosen according to the norms of scaled modes (or actually to the absolute values of the amplitudes).
Figure~\ref{pic_karman_reconst_and_cluster} (left) shows the reconstruction error over time (i.e., $\lVert x_k - \lambda^k_\mathcal{C}a_\mathcal{C}\vartheta_\mathcal{C} \rVert$ plotted for $k=0,\dots,m$) for three examples of subsets $\mathcal{C} \in \{\mathcal{C}_I,\mathcal{C}_{II},\mathcal{C}_{III}\}$.
The non-redundant eigenvalues belonging to the three subsets are illustrated in Figure~\ref{pic_karman_reconst_and_cluster} on the right  (compare Figure~\ref{pic_karman_new_ew_dom}).
The subset $\mathcal{C}_{III}$ contains all components and is the maximally achievable order of accuracy
(the reconstruction is exact, since the conditions of Equation~\ref{eq:theroem} are satisfied).
The other two subsets $\mathcal{C}_I$ and $\mathcal{C}_{II}$ contain the components with the $k$ highest influence, where different thresholds have been chosen.
However, this traditional approach for the selection of components neither clarifies the interplay of the chosen components nor classifies them appropriately.
In addition, due to the superposition principle, an incoherently selection may lead to components that eliminate each other.
In sum, incoherently selected (scaled) components may not describe any relevant spatial and temporal features of the flow.

Another approach for the selection of components can be performed on the basis of the improved components and visualizations. 
Since temporal aspects are encoded in the visualization of the dominance structure, it probably provides a better tool for the selection of components.
In general, it can be conducted in the following way (where we always select only non-redundant components): 
First, the components with a high influence at a certain time step $k_0$ are selected. 
This is achieved by selecting every component whose value $\lVert \lambda_j^{k_0} a_j \vartheta_j \rVert$ is higher than a chosen fixed threshold. 
If $k_0 = 0$, then the procedure is equal to the traditional one.
In a second step, undesirable components are sorted out manually.
Using the new dominance structure, non-selected components that represent important steady state parts can be added.
In contrast, selected components that vanish extremely fast can be excluded. 
With the traditional visualization (Figure~\ref{pic_karman_old_ew_dom} on the left), this is not possible, even if the eigenvalue visualization (Figure~\ref{pic_karman_old_ew_dom} on the right) is consulted additionally.
To assist this procedure, we keep attention to the reconstruction error over time (see Figure~\ref{pic_karman_reconst_and_cluster}) using the selected ones as well as all modes as a reference.
This shows the precision of the chosen components and helps substantiate the selection of components.

The proposed approach based on our visualizations is useful for a first impression or an explorative analysis.
Our experiences have shown that it is expedient for simpler data sets.
For complex systems that exhibit several different time-dependent phenomena, the interplay, interpretation, and classification of selected components still may remain unclear.  
So far, the selected (scaled) modes were classified by the location of their corresponding eigenvalues.
This can lead to the fact that components eliminate each other for certain sections of the flow. 
Therefore, an appropriate criterion for the selection of components is to classify those which represent a certain section of the flow accurately.
Mathematically, a discrete optimization problem can be formulated: For a certain time section $0 \leq k_1 < \dots, k_2 \leq m$, we look for a subset $\mathcal{C}_M$ with $0<M<r_0$ elements that satisfies
\begin{equation} \label{eq:measure}
\min_{\mathcal{C}_M} \sum_{k = k_1}^{k_2} \lVert x_k - \lambda^k_\mathcal{C}a_\mathcal{C}\vartheta_\mathcal{C} \rVert_2^2~.
\end{equation}
This technique can be obviously used for the selection of components, since each computed subset may characterize a section of the flow and the components its features.
For (almost) periodic data sets, the segmentation of the flow into sections is irrelevant, however, a classification into different base frequencies is of importance.
The discrete optimization problem (Equation~\ref{eq:measure}) can practically not be solved as the computational complexity grows with $\binom{r_0}{M}$ (for non-rank-deficient data $\binom{m}{M}$) for the selection of $M$ components.

Instead of using this time-consuming (PCA-like energy sorted) selection technique, we propose two fast clustering approaches for aggregating components that will detect the same components (as we will show later).
As shown above, the aggregated components reveal relevant physical features classifying the flow into segments such as the transient response or the steady state.
Subsequently, the components can be investigated individually with regard to the specific detected feature.

The following two clustering approaches operate on the eigenvalues and their new representation, i.e., we aggregate on the basis of temporal patterns.
Therefore, every aggregation include the complex conjugate counterpart of a components such that the temporal development should be evaluated using the real part as in Equation~\ref{eq:combining_counterparts}.

\paragraph{\textbf{Distance-Based Clustering}}
Several flow phenomena like damping processes are characterized by modes that exhibit very similar frequencies and growth rates. 
Therefore, we propose aggregating components with closely located eigenvalues.
For this procedure, we first thin out the potential eigenvalues such that the flow can be  reconstructed adequately by them.
Then, a distance-based clustering method is applied to those eigenvalues.
The methodology is highlighted in Figure~\ref{pic_karman_reconst_and_cluster} by (I).
It often leads to multiple clusters $C_{11} ,C_{12}, ...$ that may reveal features.
However, the eigenvalue $\lambda =1$ is always added manually to each cluster.
To check the relevance of a cluster, the temporal development is consulted.
For instance, the reconstruction error over time of cluster $C_{11}$ is shown in Figure~\ref{pic_overview} on the right.

\paragraph{\textbf{Harmonic Clustering}}
For the identification of temporal patterns, DFT uses harmonics, i.e., multiples of frequencies.
We adapt this concept and aggregate components that exhibit patterns of harmonics.
The clustering approach uses again an appropriately thinned out set of eigenvalues whose components reconstruct the flow adequately.
Then, we search for multiples in the eigenvalue representation, either manually or algorithmically.
The choice depends on the distribution and complexity of eigenvalues.
The process is demonstrated in Figure~\ref{pic_karman_reconst_and_cluster} by (II) (or Figure~\ref{pic_superposed_quadgyre}).
If there is more than one cluster, the eigenvalue $\lambda = 1$ should not be included in the cluster as constant parts are always mixed and not separable.
However, the dynamic behavior is separated by this approach, which is the key point. 
For instance, the reconstruction error over time of cluster $C_{21}$ is shown in Figure~\ref{pic_overview} on the right.

In sum, we propose using the two clustering approaches for the selection of components.
A selection consists of united clusters (found from the clustering approaches), each reveals features of the flow by inspecting the respective individual components.
This is due to the fact that the aggregations segment the flow into relevant sections and therefore classify the components.
In Figure~\ref{pic_overview}, the two clustering approaches are illustrated representing different phenomena of the flow.

\section{Results}
To demonstrate the usefulness of our proposed techniques, we apply them to unsteady flow fields to identify different frequency-based features.
The first example is a generated synthetic flow, called superposed quadgyre.
It is an overlay of two artificial 2D flow fields, called quadgyre, which are extensions of the double gyre~\cite{SHADDEN:2005:LCS} flow field
This scenario demonstrates the application of DMD to periodic flows and illustrates the correctnesss of our aggregation approaches.
The next example is a simulated von Karman vortex street that is a more complex unsteady flow resulting in an equilibrium state.
Our improved techniques allows us to identify the relevant components that are associated to the transient response and steady state.
Finally, we consider a 3D von Karman vortex street to show how the approach carries over to 3D.

\begin{figure}[t]
	\centering
	\includegraphics[width=0.9\columnwidth]{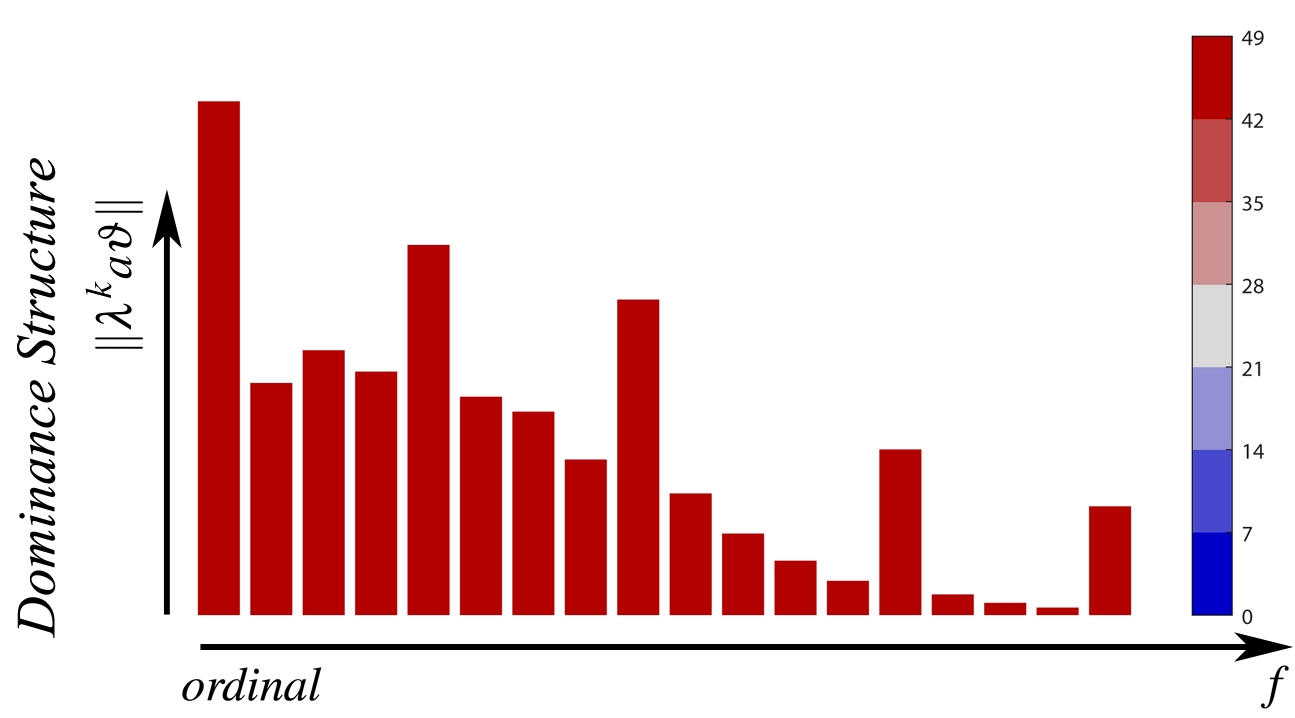}
	\caption{
	The dominance structure of the superposed quadgyre dataset. 
	Besides of decreasing behavior for higher frequencies, we immediately observe no converging and diverging parts.
	We conclude the periodicity of the flow.
	Moreover, the representation indicates the existence of two different decay patterns.
	}
	\label{pic_quadgyre_new_dom} 
\end{figure}

\subsection{Example 1: Superposed Quadgyre}
When applying DMD to overlapping periodic phenomena, such as the superposed quadgyre data set, we obtain characteristic features with DMD.
Our improved techniques identify these and the clusters represent each individual periodic phenomena as we will demonstrate in the following:
The superposed quadgyre data set is a superposition of two unsteady flow fields with different base movements and frequencies.
A full period of each flow is illustrated in Figure~\ref{pic_superposed_quadgyre} (left).
The analytical formula of both flows is given by
\begin{align*}
    u(x,y,t) &= -\pi A \sin{(\pi f(x,t))} \cos{(\pi g(y,t))} \frac{dg}{dy}(y,t)~, \\
    v(x,y,t) &=  \pi A \cos{(\pi f(x,t))} \sin{(\pi g(y,t))} \frac{df}{dx}(x,t)~,
\end{align*}
where $f(x,t) = \varepsilon \sin{(\omega_f t + s_f)}x^2 + x-2 \varepsilon \sin{(\omega_f t + s_f})x$  and $g(y,t) = \varepsilon\sin{(\frac{\omega_g}{2} t + s_g)}y^2 + y -2 \varepsilon \sin{(\frac{\omega_g}{2} t + s_g)}y$.
Both data sets consist of four vortices that move periodically. 
The vortices in the first dataset move right to left and back with parameters $A = 1$, $\varepsilon = 1$, $\omega_f = \frac{8}{27}\pi$, $\omega_g = 0$, $s_f = \frac{1}{5}\pi$, and $s_g = 0$.
The vortices from Flow II move from top right to  bottom right in a crescent-shaped motion with parameters $A = 1$, $\varepsilon = 1$, $\omega_f = \omega_g = \frac{1}{2}\pi$, and $s_f = s_g = \frac{1}{3}\pi$.
The flow fields are sampled at a resolution of 201 $\times$ 201 cells (resulting in a snapshot dimension of $n = 80802$) and have 50 time steps.

To analyze the flow decomposition, we first inspect the eigenvalue plot (Figure~\ref{pic_superposed_quadgyre}, center).
We observe that all eigenvalues have an absolute value of one, i.e., the flow is periodic and there are no converging or diverging phenomena.
This is confirmed by the dominance structure (Figure~\ref{pic_quadgyre_new_dom}) as all bars are colored red. 
The next step is to select relevant components. 
As mentioned before, for periodic systems, it is more of importance to classify the data into different base frequencies.
Hence, instead of using the traditional dominance-based approach for the selection of modes, we use the proposed harmonic clustering approach.
Two clusters $C_I$ and $C_{II}$ are detected accurately as demonstrated in Figure~\ref{pic_superposed_quadgyre}.
To verify the relevance and correctness of the two found aggregations, where each belong to a different base frequency phenomena, we compare each of them with the suitable analytical base flow.
More precisely, we consider the error between the temporal development of an aggregation (without the constant part) and the respective base flow, where we subtracted the mean.
The error plot is depicted in Figure~\ref{pic_quadgyre_reconst_2datasets}.
It can be observed that the relative error is approximately 4\% for both clusters.
Hence, each found cluster characterizes one of the base flows (which form the superposed quadgyre by superposition).
This shows that our clustering approach can extract and classify overlapping phenomena with different frequency patterns.
Therefore, the components of the two aggregations can now be analyzed separately from each other.

\begin{figure}[t]
	\centering
	\includegraphics[width=\columnwidth]{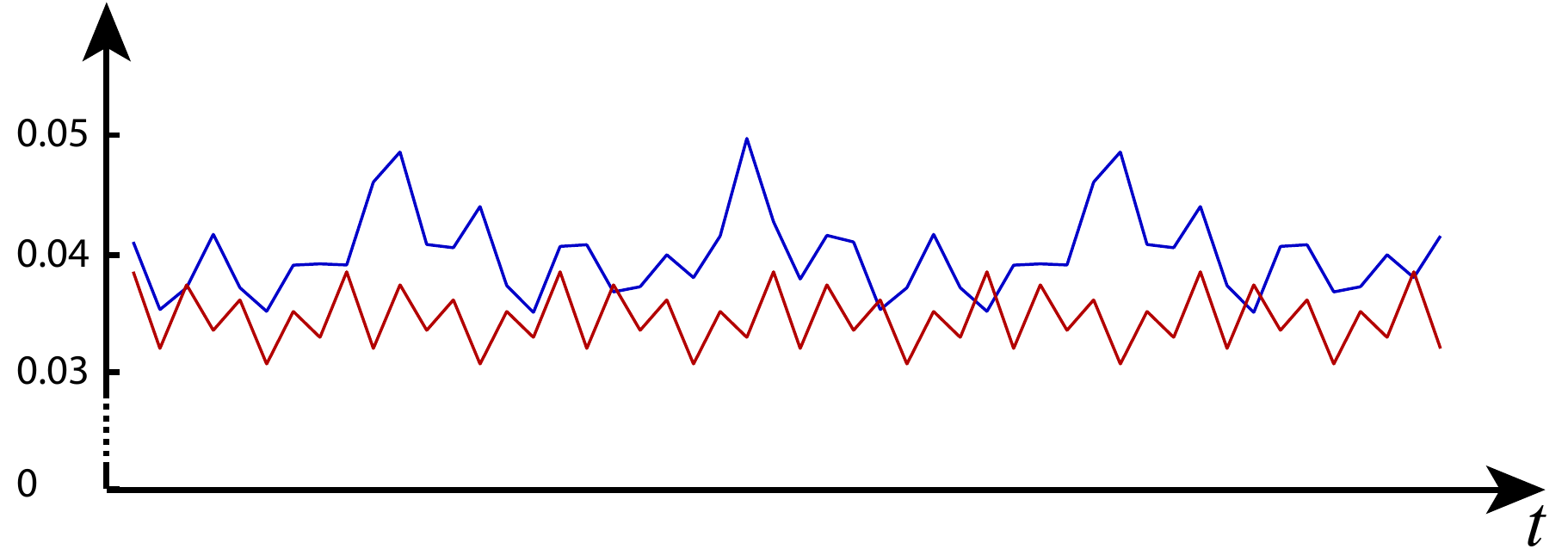}
	\caption{
	The plot shows the error between the temporal development of the two aggregations found by the harmonic clustering approach applied to the superposed quadgyre data set and the respective original flows, which are mean subtracted in order to eliminate constant parts (compare Figure~\ref{pic_superposed_quadgyre}).
	Since the error is very small, each cluster captures the dynamic behavior of one original flow accurately.
	}
	\label{pic_quadgyre_reconst_2datasets} 
\end{figure}

In Figure~\ref{pic_superposed_quadgyre} (right), the three most important modes of each cluster are depicted as well as their temporal development.
The upper aggregation belongs to the standard quadgyre.
A similar flow was investigated by Brunton et al. \cite{Brunton:2015:CompressedDMD}, however, our approach suppresses the negligible complex conjugate scaled modes.
The upper three modes show symmetric vanishing and recurring vortices indicating a periodic flow in the horizontal direction.
The lower three modes indicate a crescent shape movement of vortices.
The temporal development of the first mode shows the two vortices in the top and bottom right as well as the emerging vortex on the left. 
This vortex can be found in the other two modes as well. 
Whereas the second one mainly describes the crescent-shaped movement (since the frequency is twice as the base frequency), the last one highlights the vortices in the left upper and lower corners.
The temporal evolution of the modes (Figure~\ref{pic_superposed_quadgyre} right) shows that the symmetry properties of both systems are conserved.

If the DFT is consulted for such a frequency-based investigation, a classification like this would not work as the frequencies cannot be identified directly.
As the DFT uses uniformly distributed frequencies depending on the total number of snapshots, the mixed frequencies of the superposed quadgyre will not be determined.
Therefore, the frequency detection with DFT is blurred and none of the original flows is precisely detected.

\begin{figure}[t]
	\centering
	\includegraphics[width=\columnwidth]{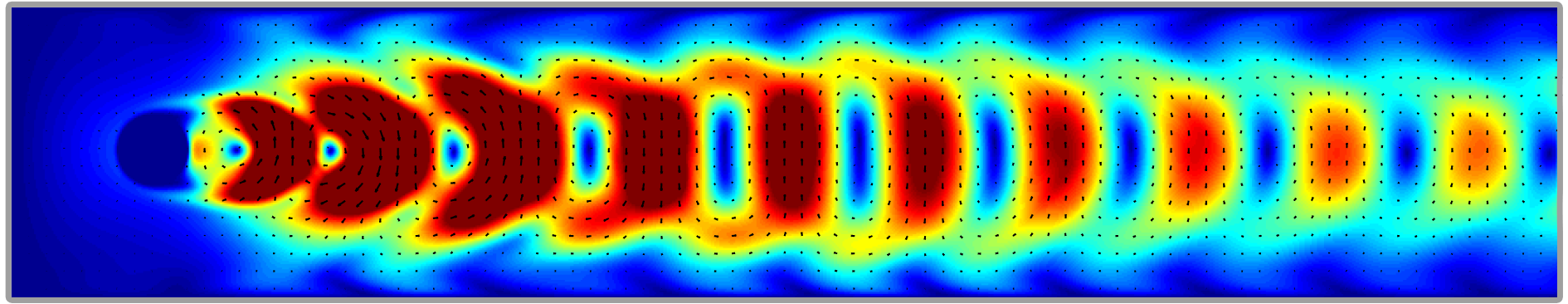}	\includegraphics[width=\columnwidth]{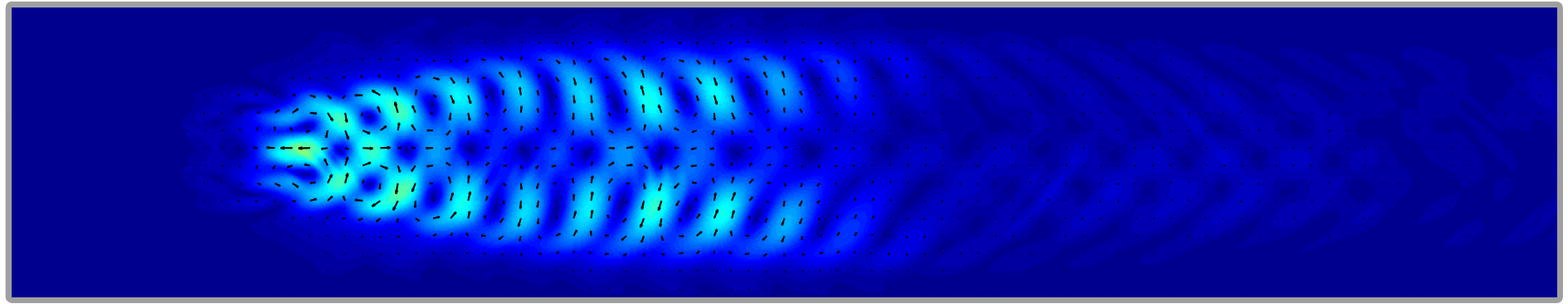}	\includegraphics[width=\columnwidth]{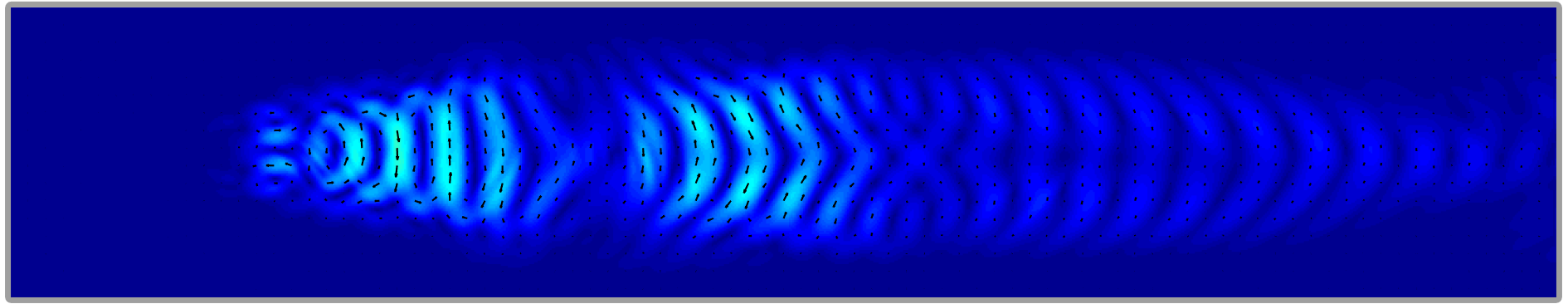}	\includegraphics[width=\columnwidth]{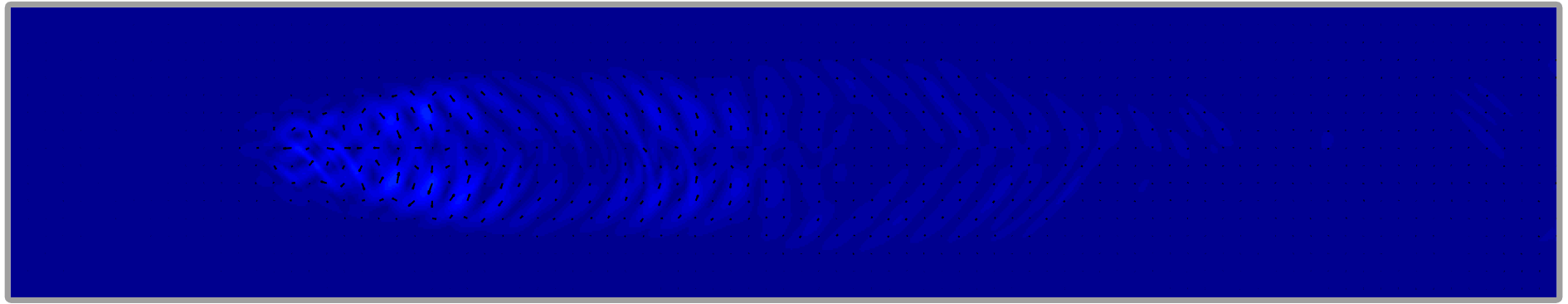}
	\caption{
		Visualization of scaled modes from the 2D Karman vortex street.
		These are selected from an aggregation that represent the steady state.
		The aggregation was found by the harmonic clustering approach.
	}
	\label{pic_karman_dmd_modes}
\end{figure}

\begin{figure}
	\centering
		\includegraphics[width=\columnwidth]{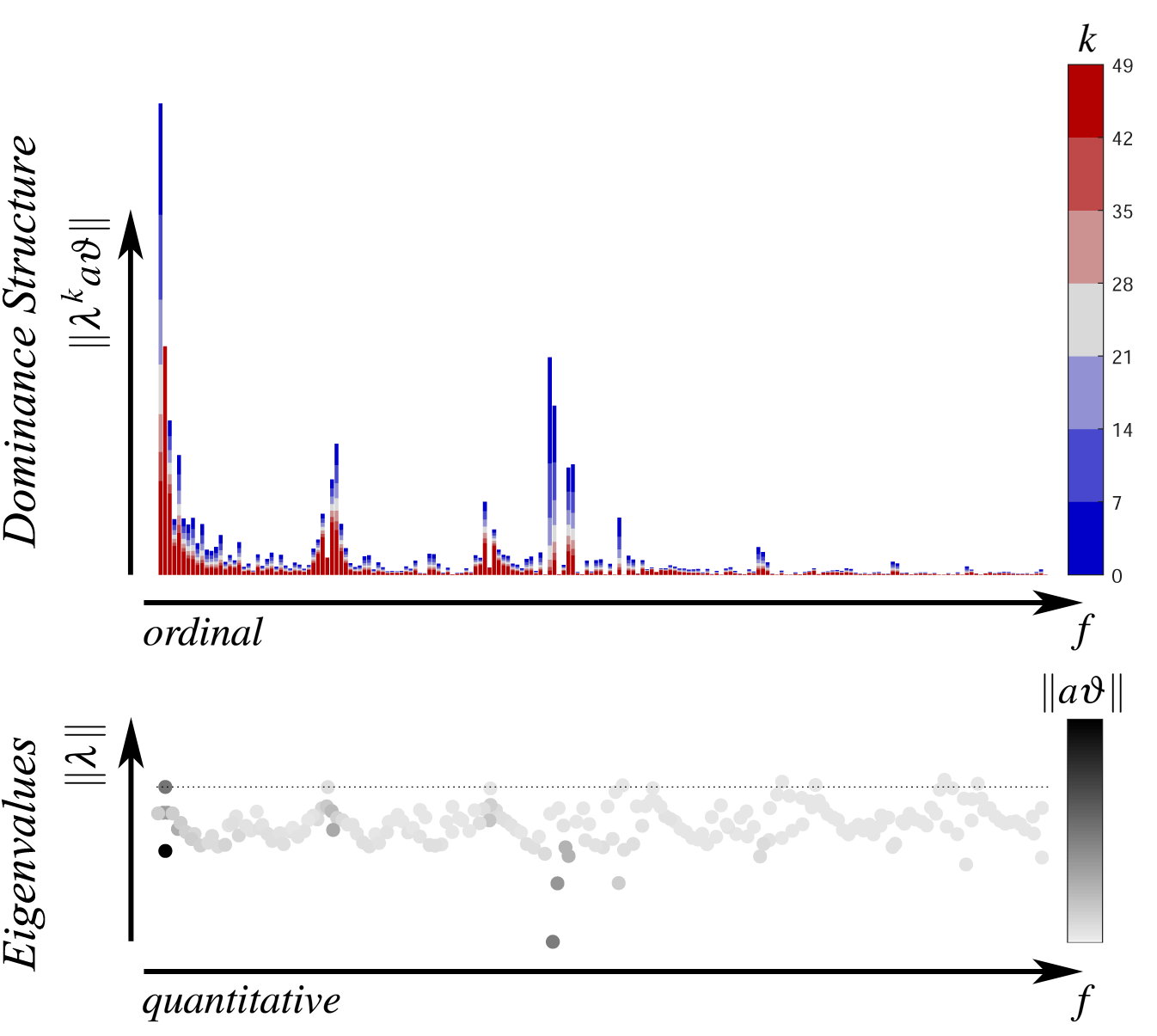}
	\caption{
		Our visualizations of the dominance structure and eigenvalues for the 3D Karman vortex street. 
		We observe the identical patterns as in the 2D counterpart.
		}
	\label{pic_karman3d_ew_dom} 
\end{figure}

\subsection{Example 2: 2D Karman Vortex Street}
To compare our improved components and visualizations with the traditional ones, a flow past a circular object forming a von Karman vortex street is investigated (see Figure~\ref{pic_overview} on top). 
This typical analysis example for DMD was simulated on a grid with 881 $\times $ 166 points.
Applying DMD to the 251 snapshots results in 250 DMD components.
In Figure~\ref{pic_karman_new_ew_dom}, the dominance structure and eigenvalues are visualized with our new approach.
In this example, we are faced with a rather complicated dominance structure.
It can be observed that a lot of scaled modes have a large influence on the entire system.
However, we can also immediately observe that nearly all components characterize damping phenomena, in particular those with a high impact at the beginning.
After about 50 time steps, many components have vanished and parts of the steady state (highlighted by pure red bars) will become more dominant.
Hence, by the temporal encoding in the dominance visualization, we get a precise understanding of the component's impact.
In addition, as the two axes of the dominance and eigenvalue visualization are linked, we observe a similar wave-shaped contour.
The traditional visualizations in Figure~\ref{pic_karman_old_ew_dom} do not detect or highlight any of these aspects.
Based on the observations of our improved visualizations, it may be possible to make an adequate selection of components, however, we want to demonstrate the usefulness of our cluster approaches for the selection of components.

In Figure~\ref{pic_karman_reconst_and_cluster}, the reconstruction error over time for different subsets of components is shown where the corresponding eigenvalues are represented on the right.
For the aggregation of components, we use both clustering approaches on these different representations that facilitate the detection of patterns.
The whole procedure is demonstrated in Figure~\ref{pic_karman_reconst_and_cluster} and the two chosen clusters $C_{1,1}$ and $C_{2,1}$ segment the flow into physically relevant parts.
Figure~\ref{pic_overview} shows the errors over time of the components belonging to the two aggregations C1,1 and C2,1.
For the first aggregation, we observe a small error at the beginning of the simulation.
Therefore, the transient response can be described by the individual components of C1,1, whereas aggregation C2,1 characterizes the steady state as the error decreases over time.
Thus, we have detected two important phenomena that classify the components.

A comprehensive time-continuous and -discrete analysis of wake flows with DMD was conducted by Bagheri \cite{bagheri:2013:KoopamnCylinderWake}.
In this context, a classification was achieved that equals the result of our cluster approaches.
To validate the cluster approaches once again, we computed the discrete minimization problem in Equation~\ref{eq:measure}.
We applied it both to a section in the beginning and in the end, where we search for a subset $\mathcal{C}_M$ with $M=4$ components.
The resulting components correspond to the ones found by the cluster approaches.
Consequently, these components are the best selection for the representation of the transient response in the beginning and the steady state in the end.

Tu et al.\cite{tu:2014:on_dmd_theorey_and_app} combine multiple experiments to detect harmonics of a von Karman vortex street.
They showed that in a simple experiment the patterns are harder to identify.
However, using our improved components and visualizations as well as the harmonic clustering approach, the patterns can be detected easily.
For the further analysis of these phenomena, we restrict the investigation to the components within these two aggregations.

The scaled mode as well as the traditional mode corresponding to the eigenvalue $\lambda=1$ (which is contained in both clusters) are compared in Figure~\ref{pic_dmd_mode_EW1_compare}. 
As mentioned before, the scaled mode represents the flow accurately, unlike the traditional mode.
Figure~\ref{pic_karman_dmd_modes} shows the four scaled modes from cluster C2,1. 
These components represent the characteristic Karman vortices by superposition.
In accordance with the sorting, the scaled modes reveal more and more fine scale patterns.
This is due to the fact that fine scale patterns are represented by higher frequencies.
Moreover, the four scaled modes are absolutely symmetrical (with regard to the y-axis), whereas the traditional modes do not exhibit these structures \cite{tu:2014:on_dmd_theorey_and_app}.

\subsection{Example 3: 3D Karman Vortex Street}
DMD is independent of the dimensionality of the spatial domain. 
For an analysis with DMD, we just need to adjust the visualization of the scaled modes to the dimensionality of the spatial domain.
Therefore, our investigation approach with DMD can be equally applied to 3D unsteady flow, except for the computational time.
To demonstrate this, we analyze a 3D von Karman vortex street representing an extension to the 2D counterpart. 
The grid has a resolution of $228 \times 44 \times 44$ and 381 snapshots were considered.
The flow field starts analogously as the 2D scenario in an emerging phase of the vortex street.

In Figure~\ref{pic_karman3d_ew_dom}, the dominance structure and eigenvalue visualization are depicted.
In both visualizations, we observe the same structure as in the 2D counterpart in Figure~\ref{pic_karman_new_ew_dom}.
An evaluation can be indeed conducted analogously.

For the investigation of spatial properties, we exemplarily visualize a mode in Figure~\ref{pic_karman3d_dmd_mode}.
Our visualization for the 3D scaled mode does not aim to highlight features, it is used to show similarities with the 2D counterpart.
We use a volume based representation as well as a cross section that shows the interaction inside.
The cross section exhibits similar structures as the modes from the 2D counterpart (Figure~\ref{pic_karman_dmd_modes}).

For 3D data sets, the computation time is affected moderately as only the SVD of the matrix $X$ and the computation of the amplitudes scales linearly with the additional dimension.
Furthermore, our visualizations and clustering approaches are only based on multiple calculations with the DMD eigenvalues.
Therefore, no overhead is produced.

\begin{figure}[tb]
	\centering
	\includegraphics[width=\columnwidth]{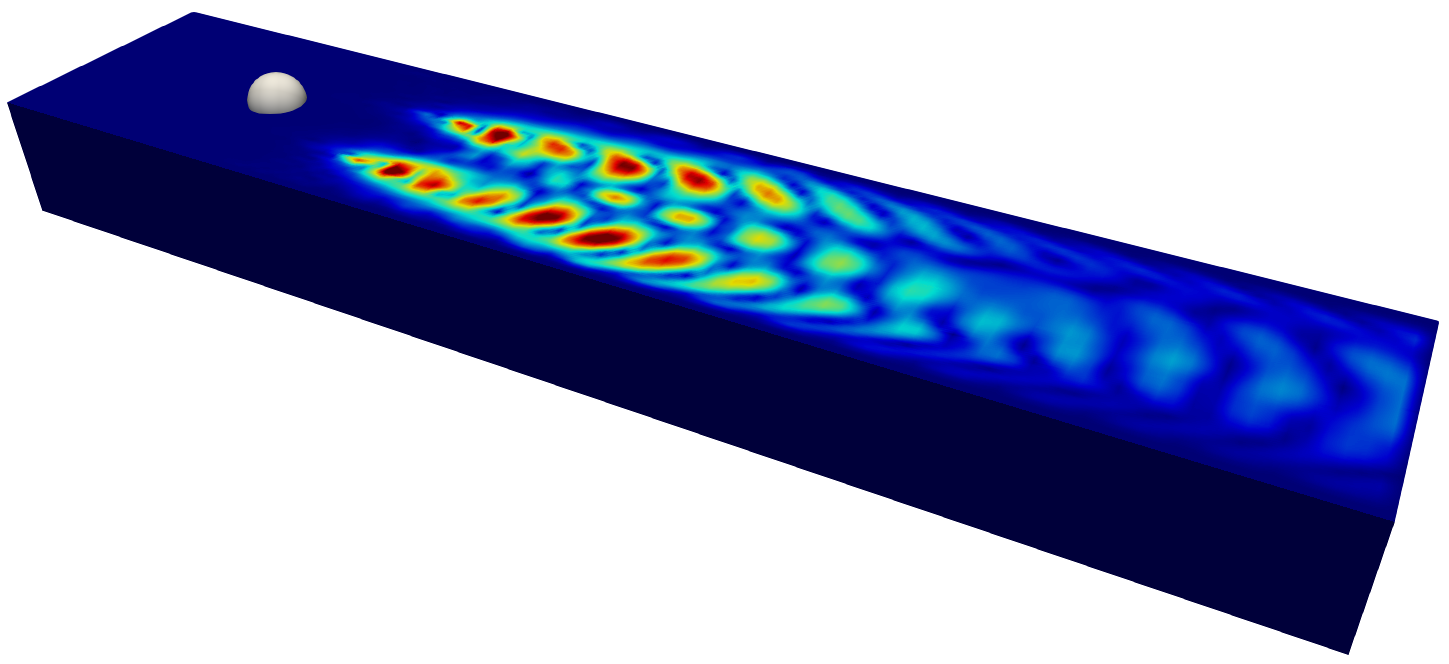}
	\caption{
		Visualization of a scaled mode for the 3D von Karman vortex street. 
		A volume-based representation with a cross section is used to show the connection to the 2D case.
		The frequency of the second 2D mode in Figure~\ref{pic_karman_dmd_modes} is approximately equal and, therefore, the same spatial behavior can be observed in the cross section.
		}
	\label{pic_karman3d_dmd_mode}
\end{figure}

\section{Conclusion}
In this paper, we have thoroughly recapped the mathematical foundation of DMD that allows us to combine and improve the DMD components such that the underlying physics is represented more adequately, e.g., by the construction of scaled modes.
Moreover, the interplay of the components is clarified and a new view on DMD is given by comparing it to DFT.
These new insights are used to design appropriate visualizations such that the spatio-temporal character of DMD is respected and redundant parts are hidden.
Therefore, a more adequate selection of components can be made and the identification of specific patterns and features is facilitated.
The two novel clustering approaches should be additionally consulted for the selection of components as these segment the flow into physically relevant sections.
In sum, a deeper understanding of the DMD components is gained that make DMD more accessible for users.

These new techniques may also be combined with other visualization methods that highlight further features in the components.
For instance, we could apply methods like FTLE to the temporal development of scaled modes.
In this case, FTLE could be seen as a post-processing step to gain another view on the spatial components.
Another interesting research direction might be the application of our techniques to other data.
So far, we only used data represented by a grid that consists of velocity components. 
It would be interesting to choose a different data representation, like particle-data, and evaluate different quantities, e.g., pressure or vorticity.
We plan to apply our proposed techniques to other (real-valued) data and evaluate the effectiveness for different applications again.
For example, the clustering approaches are only useful, if the data exhibits similar frequency patterns such as fluid flow.

\acknowledgments{
This work is partly supported by “Kooperatives Promotionskolleg Digital Media” at Hochschule der Medien and the University of Stuttgart.}

\bibliographystyle{abbrv-doi}
\bibliography{paper}

\section*{Appendix: Complex Conjugated DMD Components} \label{sec:appendix}

\begin{lemma} 
Let $\beta \in \mathbb{C}$ and $z \in \mathbb{C}^n$ with $\Re(z)$ and $\Im(z)$ are linearly independent. 
Then, the following assertion holds:
\begin{equation*}
	\beta \neq 1 \implies z + \beta \overline{z} \notin \mathbb{R}.
\end{equation*}
\end{lemma}

\begin{proof}
First, the imaginary part of $z + \beta \overline{z}$ is calculated:
\begin{align*}
	z + \beta \overline{z} = &\Re(z) + i \Im(z) + (\Re(\beta) + i \Im(\beta)) (\Re(z) - i \Im(z)) \\
	= &\Re(z) + i \Im(z) + \Re(\beta) \Re(z) + \Im(\beta) \Im(z) \\
	&+ i \Im(\beta) \Re(z) - i \Re(\beta) \Im(z) \\
	= &[(1 + \Re(\beta))\Re(z) + \Im(\beta) \Im(z)] \\
	&+ i[\Im(\beta) \Re(z) + (1 - \Re(\beta)) \Im(z)].
\end{align*}
The proof is done by contraposition, i.e., let us assume that $z + \beta \overline{z} \in \mathbb{R}$. 
Then, the imaginary part has to equal zero, i.e.,
\begin{equation*}
	\Im(\beta_j) \cdot \Re(z_j) + (1 - \Re(\beta_j)) \cdot \Im(z_j) = 0.
\end{equation*}
Since $\Re(z)$ and $\Im(z)$ are linearly independent, we conclude that
	\begin{equation*}
	1-\Re(\beta_j) = 0~, \qquad \Im(\beta_j) = 0~,
	\end{equation*}
which implies $\beta_j = 1$. 
This proves the statement by contraposition.
\end{proof}

\begin{theorem}
Consider real-valued data $x_0,x_1,\dots,x_m \in \mathbb{R}^n$.
If both $x_0,\dots,x_{m-1}$ and $x_1,\dots,x_m$ are linearly independent and the DMD eigenvalues $\lambda_1,\dots,\lambda_m$ are distinct, then Algorithm~\ref{Algorithm_DMD} produces pairs of complex conjugate eigenvalues, i.e., if $\lambda \in \mathbb{C}\setminus \mathbb{R}$ is a pure complex DMD eigenvalue then the complex conjugate $\overline{\lambda}$ is a DMD eigenvalue, too, and the associated scaled DMD modes are given by $a \vartheta$ and $\overline{a \vartheta}$, respectively.
\end{theorem}

\begin{proof} 
We use the notation from Algorithm~\ref{Algorithm_DMD}. 
Since we consider real-valued data, the matrices $X,Y \in \mathbb{R}^{n \times m}$ will be real-valued.
For real-valued matrices the singular value decomposition can be chosen real-valued and hence the DMD matrix $S \in \mathbb{R}^{m \times m}$ is real-valued.
As a result, the eigenvalues of $S$ (which are the DMD eigenvalues) occur in complex conjugate pairs. 
Let $\lambda_1, \lambda_2 = \overline{\lambda_1},\lambda_3, \lambda_4 = \overline{\lambda_3}, \dots, \lambda_{2k-1}, \lambda_{2k} = \overline{\lambda_{2k-1}} \in \mathbb{C}\setminus \mathbb{R}$ be the complex conjugate pairs and $\lambda_{2k+1},\dots,\lambda_m \in \mathbb{R}$ the remaining real-valued eigenvalues. 
The corresponding eigenvectors are given by $v_1, v_2, \dots, v_{2k-1}, v_{2k}, v_{2k+1}, \dots, v_m$.
Since the eigenvalues are distinct, which implies a one-dimensional eigenspace, and occur in complex conjugate pairs, the following relations hold
\begin{equation*}
    v_2 = c_1 \overline{v_1},
    \qquad \dots \qquad
    v_{2k} = c_{2k-1} \overline{v_{2k-1}},
\end{equation*}
for some scaling factors $c_1, c_3, \dots, c_{2k-1} \in \mathbb{C}$ with $\lvert c_1 \rvert = \dots = \lvert c_{2k-1} \rvert = 1$.
	
The DMD modes are calculated by $\vartheta_j = \frac{1}{\lambda_j} Y V \Sigma^{-1} v_j$. 
Hence, the DMD modes maintain the structure of the eigenvectors, i.e., we can denote those analogously by $\vartheta_1,  \vartheta_2, \dots, \vartheta_{2k-1}, \vartheta_{2k}, \vartheta_{2k+1} \dots, \vartheta_m$ with
\begin{equation*}
	\vartheta_2 =  c_1 \overline{\vartheta_1},
	\qquad \dots \qquad
	\vartheta_{2k} = c_{2k-1} \overline{\vartheta_{2k-1}}.
\end{equation*}
	
Using the reconstruction property, i.e., Equation~\ref{eq:theroem} for the second time step, we get the following relationship:
\begin{align*}
	x_1
	&= \sum_{l=1}^{2k} \lambda_l a_l \vartheta_l 
	+  \sum_{l=2k+1}^m \lambda_l a_l \vartheta_l \\
	&= \sum_{l=1}^{k} \lambda_{2l-1} a_{2l-1} \vartheta_{2l-1} 
	+  \sum_{l=1}^{k} \lambda_{2l  } a_{2l  } \vartheta_{2l  } 
	+  \sum_{l=2k+1}^m \lambda_l a_l \vartheta_l \\	
	&= \sum_{l=1}^{k} \lambda_{2l-1} a_{2l-1} \vartheta_{2l-1} 
	+  \sum_{l=1}^{k} \overline{\lambda_{2l-1}} a_{2l} c_{2l-1} \overline{\vartheta_{2l-1}} 
	+  \sum_{l=2k+1}^m \lambda_l a_l \vartheta_l,
\end{align*}
where $a_1, a_2, \dots, a_m$ are the DMD amplitudes. 
As we assume real-valued data, which implies in particular $x_1 \in \mathbb{R}^n$, the above sum has to be real-valued. 
Let us express the DMD amplitudes $a_2, a_4 \dots, a_{2k}$ (belonging to the complex conjugate counterpart) as 
\begin{equation*}
    a_2 = b_1 \overline{a_1}, 
	\qquad 	\dots 	\qquad
	a_{2k} = b_{2k-1} \overline{a_{2k-1}},
\end{equation*}
for some appropriate scaling factors $b_1,b_3\dots,b_{2k-1} \in \mathbb{C}$. 
If we define $\beta_1 = b_1 c_1,\beta_3 = b_3 c_3, \dots, \beta_{2k-1} = b_{2k-1} c_{2k-1} \in \mathbb{C}$, we can express the above sum as
\begin{equation*}
       \sum_{l=1}^{k} \lambda_{2l-1} a_{2l-1} \vartheta_{2l-1} 
	+  \sum_{l=1}^{k} \beta_{2l-1} \overline{\lambda_{2l-1} a_{2l-1} \vartheta_{2l-1}}
	+  \sum_{l=2k+1}^m \lambda_l a_l \vartheta_l 
	\in \mathbb{R}^n.
\end{equation*}
Now, the proof is complete, if we are able to show that $\beta_1 = \beta_3 = \dots = \beta_{2k-1} = 1$, since every scaled DMD mode $a_{2l} \vartheta_{2l}$ (belonging to a DMD eigenvalue with a complex conjugate pair of eigenvalues) is given by
	\begin{equation*}
	a_{2j} \vartheta_{2j}
	= b_{2j-1} c_{2j-1} \overline{a_{2j-1} \vartheta_{2j-1}}
	= \beta_{2j-1} \overline{a_{2j-1} \vartheta_{2j-1}}
	= \overline{a_{2j-1} \vartheta_{2j-1}}.
	\end{equation*}
	
To prove the statement, consider the real-valued sum from above, however, for simplicity we use $z_j = \lambda_j a_j \vartheta_j$:
\begin{equation*}
       \sum_{l=1}^{k} z_{2l-1}
	+  \sum_{l=1}^{k} \beta_{2l-1} \overline{z_{2l-1}}
	+  \sum_{l=2k+1}^m z_l 
	\in \mathbb{R}^n.
\end{equation*}
Assume that there is at least one coefficient $\beta_{l_0} \neq 1$. 
Since the DMD modes are linearly independent (because the snapshots $x_1,\dots,x_m$ are linearly independent), the vectors $\vartheta_1, \overline{\vartheta_1}, \dots, \vartheta_{2k-1}, \overline{\vartheta_{2k-1}}, \vartheta_{2k+1}, \dots, \vartheta_m$ are linearly independent. 
In addition, a simple calculation shows that $\Re(\vartheta_1), \Im(\vartheta_1), \Re(\vartheta_3), \Im(\vartheta_3), \dots, \Re(\vartheta_{2k-1}), \Im(\vartheta_{2k-1}), \vartheta_{2k+1}, \dots, \vartheta_m$ are linearly independent as well.
Finally, $\Re(z_{j_0})$ and $\Im(z_{j_0})$ are linearly independent and by the previous proven Lemma, we conclude that $z_{j_0} + \beta_{j_0} \overline{z_{j_0}} \notin \mathbb{R}^n$.
However, there is no possibility to eliminate the upcoming imaginary part, though, as the full sum has to be real-valued. 
Consequently, the assumption is wrong and $\beta_1 = \dots = \beta_{2k-1} = 1$, which completes the proof.
\end{proof}

\end{document}